\documentclass[a4paper,11pt]{article}
\usepackage{mathtools}
\usepackage{amsmath}
\usepackage{amsthm}
\usepackage{amssymb}
\usepackage[T1]{fontenc}
\usepackage{fullpage}
\usepackage{graphicx}
\usepackage{xspace} 
\usepackage{hyperref}
\usepackage{pdfpages}
\usepackage[absolute,overlay]{textpos}

\usepackage{chngcntr}
\usepackage{apptools}
\AtAppendix{\counterwithin{lemma}{section}}
\usepackage{multido}
\usepackage{verbatim}

\newtheorem{lemma}{Lemma}

\newtheorem{corollary}{Corollary}
\newtheorem{theorem}{Theorem}

\newcommand{\myparagraph}[1]{\medskip\noindent{\textbf{#1}}}

\newcommand{\ccheck}[1]{\textcolor{red}{#1}}

\newcommand{\candset}[1]{\ensuremath{\mathsf{T}_{#1}}\xspace}
\newcommand{\one}{\ensuremath{\mathsf{1}}\xspace}
\newcommand{\onec}{\ensuremath{\mathsf{1C}}\xspace}
\newcommand{\oner}{\ensuremath{\mathsf{1R}}\xspace}

\newcommand{\twoc}{\ensuremath{\mathsf{2C}}\xspace}
\newcommand{\twor}{\ensuremath{\mathsf{2R}}\xspace}
\newcommand{\twocc}{\ensuremath{\mathsf{2Cc}}\xspace}
\newcommand{\twoch}{\ensuremath{\mathsf{2Ch}}\xspace}
\newcommand{\twoco}{\ensuremath{\mathsf{2Co}}\xspace}
\newcommand{\tworv}{\ensuremath{\mathsf{2Rv}}\xspace}
\newcommand{\twori}{\ensuremath{\mathsf{2Ri}}\xspace}
\newcommand{\card}[1]{\ensuremath{\mathsf{card}(#1)}}
\newcommand{\len}[1]{\ensuremath{\mathsf{len}(#1)}}

\newcommand{\cut}{\ensuremath{\ell}}
\newcommand{\ol}{\ensuremath{\overline}}
\newcommand{\mybd}[1]{\ensuremath{\partial #1}}
\newcommand{\myink}[1]{\ensuremath{\mathsf{Ink}(#1)}}
\newcommand{\mywidth}[1]{\ensuremath{\mathsf{Width}(#1)}}
\newcommand{\myint}[1]{\ensuremath{\mathsf{int}(#1)}}
\newcommand{\pt}{\ensuremath{\mathsf{P}}}
\newcommand{\qt}{\ensuremath{\mathsf{Q}}}
\newcommand{\rt}{\ensuremath{\mathsf{R}}}
\newcommand{\lt}{\ensuremath{\mathsf{L}}}
\newcommand{\gt}{\ensuremath{\mathsf{G}}}
\newcommand{\gtpr}{\ensuremath{\mathsf{G}'}}
\newcommand{\ut}{\ensuremath{\mathsf{U}}}

\newcommand{\rom}[1]{\uppercase\expandafter{\romannumeral #1\relax}}
\newcommand{\tpv}{\textsf{vt}-algo\xspace}

\newcommand{\mip}{\textsf{i}-partition\xspace}
\newcommand{\tp}{\textsf{t}-partition\xspace}
\newcommand{\vtp}{\textsf{vt}-partition\xspace}

\newcommand{\atp}{\textsf{at}-partition\xspace}
\newcommand{\kc}{\kappa}
\newcommand{\vcut}{\textsf{v}-cut\xspace}
\newcommand{\acut}{\textsf{a}-cut\xspace}
\newcommand{\vcuts}{\textsf{v}-cuts\xspace}
\newcommand{\acuts}{\textsf{a}-cuts\xspace}
\newcommand{\vgadget}{\textsf{v}-gadget\xspace}
\newcommand{\vgadgets}{\textsf{v}-gadgets\xspace}
\newcommand{\cgadget}{\textsf{c}-gadget\xspace}

\newcommand{\myth}{\textsf{TH}\xspace}
\newcommand{\plsat}{\textsf{P3SAT}\xspace}
\newcommand{\true}{\texttt{true}\xspace}
\newcommand{\false}{\texttt{false}\xspace}
\newcommand{\vwind}{\textsf{v}-windmill\xspace}
\newcommand{\hwind}{\textsf{h}-windmill\xspace}
\newcommand{\kcq}[1]{\ensuremath{\kappa_{\mathsf{\rom{#1}}}}}
\newcommand{\oq}[1]{\ensuremath{o_{\mathsf{\rom{#1}}}}}
\newcommand{\minrest}{\ensuremath{\mathsf{W}(p)}}
\newcommand{\qrest}{\ensuremath{\mathsf{Q}(p)}}

\newbox\ProofSym
\setbox\ProofSym=\hbox{%
	\unitlength=0.18ex%
	\begin{picture}(10,10)
	\put(0,0){\framebox(9,9){}}
	\put(0,3){\framebox(6,6){}}
	\end{picture}}

\title{Rectangular Partitions of a Rectilinear Polygon\thanks{
This research were supported by the Institute of Information \& communications
Technology Planning \& Evaluation(IITP) grant funded by the Korea government(MSIT)
(No. 2017-0-00905, Software Star Lab (Optimal Data Structure and Algorithmic Applications in Dynamic Geometric Environment))
and (No. 2019-0-01906, Artificial Intelligence Graduate School Program(POSTECH)).}}

\author{Hwi Kim\thanks{Department of Computer Science and Engineering, Pohang University of Science and Technology, 
Pohang, Korea. \texttt{hwikim@postech.ac.kr}}
\and Jaegun Lee\thanks{Department of Convergence IT Engineering, Pohang University of Science and Technology, 
Pohang, Korea. \texttt{jagunlee@postech.ac.kr}}
\and Hee-Kap Ahn\thanks{Department of Computer Science and Engineering, Graduate School of Artificial Intelligence, 
Pohang University of Science and Technology, Pohang, Korea. \texttt{heekap@postech.ac.kr}}}

\begin{document}
\date{}
\maketitle

\begin{abstract}
    We investigate the problem of partitioning a 
    rectilinear polygon $P$ with $n$ vertices and no holes 
    into rectangles using disjoint line segments drawn inside $P$ 
    under two optimality criteria. 
    In the minimum ink partition, 
    the total length of the line segments drawn inside $P$ is minimized.
    We present an $O(n^3)$-time algorithm 
    using $O(n^2)$ space that returns a minimum ink partition of $P$. 
    In the thick partition, the minimum side length over all resulting rectangles is maximized.
    We present an $O(n^3 \log^2{n})$-time algorithm using $O(n^3)$ space 
    that returns a thick partition using line segments incident to vertices of $P$, 
    and an $O(n^6 \log^2{n})$-time algorithm using $O(n^6)$ space 
    that returns a thick partition using line segments incident to the boundary of $P$.
    We also show that if the input rectilinear polygon has holes, 
    the corresponding decision problem for the thick partition problem using line segments
    incident to vertices of the polygon is NP-complete.
    We also present an $O(m^3)$-time $3$-approximation algorithm for the minimum ink partition
    for a rectangle containing $m$ point holes.
\end{abstract}

\section{Introduction}
\label{sec:introduction}
Partitioning geometric objects into disjoint parts of certain simple shapes 
(such as triangles, quadrilaterals, convex polygons, or their higher-dimensional analogues)
is a fundamental problem, frequently arising in the analyses of geometric and combinatorial 
properties of geometric objects. 
A classic and typical example is 
the triangulation of a simple polygon in the plane~\cite{avis,chazelleT}. 
Geometric structures such as 
the Voronoi diagram~\cite{voronoi2,voronoi1} and the Delaunay triangulation~\cite{Del34}
partition the underlying space into regions based on proximity. 
There are various applications in chip manufacturing~\cite{liu}, geoinformatics~\cite{niu}, and pattern recognition~\cite{avis,suk}.

The problem of partitioning a \emph{rectilinear} polygon into \emph{rectangles} 
has attained attention in
computational geometry in the last decades, due to its real-world applications, 
including VLSI layout design~\cite{lopez,rivest,sato}
and image processing~\cite{ferrari,gourley,mohamed}.

In this paper, we study the problem of partitioning an axis-aligned 
rectilinear polygon $P$ 
into axis-aligned rectangles using disjoint open line segments 
under two optimality criteria.
In the \emph{minimum ink partition} (\mip, in short), we obtain a partition of $P$ into 
rectangles such that the total length of the line segments used in the partition is 
the minimum among all partitions of $P$ into rectangles.
Since the total length of the line segments is exactly half the sum of the perimeters of
the resulting rectangles minus half the perimeter of $P$, 
this partition minimizes the total perimeter of the resulting rectangles.
In the \emph{thick partition} (\tp, in short),
we obtain a partition of $P$ into rectangles  
such that the minimum side length of the rectangles in the partition 
is the maximum among all partitions of $P$ into rectangles.
If there are two or more such partitions, we break ties among them 
by favoring one with the fewest number of rectangles.

Both problems have applications in VLSI layout design.
A typical problem is to group VLSI circuits into 
several rectangle-shaped channels such that the channel-to-channel interaction 
is minimized. 
The amount of channel-to-channel interaction is known to be
proportional to the total length of the sides incident to two different channels~\cite{lingas}.
Another problem arises in etching VLSI masks by electron beams with 
a fixed minimum width. Then the mask is required to be partitioned 
into rectangles with side lengths at least the minimum width to avoid 
unnecessary overexposure~\cite{oRourke}.

\subsection{Previous works}
Lingas et al.~\cite{lingas} are perhaps the first who proposed 
the \mip problem. They gave a sketch of an $O(n^4)$-time 
algorithm using dynamic programming for rectilinear polygons with $n$ vertices and
no holes in the plane. 
Their algorithm has some flaws and does not work correctly for certain polygons,
which can be fixed by handling missing cases without increasing the time complexity.
For a rectilinear polygon containing holes, they showed that 
the corresponding decision problem 
for the \mip problem is strongly NP-complete.
From then on, approximation algorithms have been presented~\cite{gonzalezO,levcopoulos,lingasH}.

For the special case of partitioning a \emph{rectangle} containing $m$ point holes 
into rectangles that contain no holes in their interiors,
there are approximation algorithms using divide-and-conquer~\cite{gonzalezB,levcopoulos}, 
transformation~\cite{gonzalezA}, and dynamic programming~\cite{duM}. 
There is a polynomial-time approximation scheme (PTAS) for this problem~\cite{mitchell99}.

The \tp problem was studied by O'Rourke and Tewari~\cite{oRourke}. 
They conjectured that the problem is NP-hard if holes are allowed, 
and claimed
an $O(n^{42})$-time algorithm for rectilinear polygons with $n$ vertices 
without holes. 
When the line segments of a partition are restricted to be incident 
to polygon vertices (\emph{vertex incidence}),
they gave an $O(n^{5})$-time and $O(n^4)$-space algorithm, 
by using observations similar to the ones by Lingas et al.~\cite{lingas}.
When each line segment of a partition is restricted to be incident to 
the boundary of the input polygon (\emph{boundary incidence}),
their algorithm under the vertex incidence
can be used, with some modifications, for the problem with
$O(n^{10})$ running time. 

\subsection{Our results}
For a rectilinear polygon $P$ with $n$ vertices and no holes in the plane,
we present dynamic programming algorithms improving upon
the $O(n^4)$-time algorithm by Lingas et al.
for the \mip problem and the $O(n^{5})$-time algorithm by O'Rourke and Tewari 
for the \tp problem. 
Our \mip algorithm takes $O(n^3)$ time and uses $O(n^2)$ space. 
This algorithm can be extended to a $3$-approximation algorithm with $O(m^3)$ time
for partitioning a rectangle containing $m$ point holes into rectangles containing no holes in their interiors~\cite{gonzalezB}.
Our \tp algorithm takes $O(n^3 \log^2{n})$ time and $O(n^3)$ space
under the vertex incidence,
and $O(n^6 \log^2{n})$ time and $O(n^6)$ space 
under the boundary incidence. 
Finally, we show that if the input rectilinear polygon has holes, 
the corresponding decision problem of 
the \tp problem under the vertex incidence is NP-complete.

\myparagraph{Sketches of our algorithms.}
Our algorithms are based on the work by Lingas et al.~\cite{lingas}.
Their algorithm uses \emph{cutsets} (to be defined later), each consisting of disjoint open 
axis-aligned line segments, called \emph{cuts}, that induce a partition of $P$ into rectilinear subpolygons.
For a rectangle $R$ contained in a rectilinear polygon $Q$,
their algorithm defines a \emph{uni-rectangle partition} of $Q$ for $R$
to be a partition of $Q$ by a cutset $\lt$ into $R$ and other rectilinear subpolygons 
such that each cut $\cut\in\lt$ overlaps a side of $R$.
The algorithm separates rectangles one by one recursively 
from a rectilinear subpolygon $Q$ of $P$ by a uni-rectangle partition of $Q$. 
It maintains an invariant, \emph{the 2-cut property,} that
each subpolygon, except the rectangle, obtained from a uni-rectangle partition of every subpolygon $Q$ of $P$
has exactly one boundary chain consisting of 
at most two line segments contained in the interior of $P$. 

Our \mip algorithm also maintains the 2-cut property, but handles those subpolygons
efficiently by classifying them into four types 
and enumerating uni-rectangle partitions without duplicates while guaranteeing an \mip.
By the 2-cut property, there are $O(n^2)$ subpolygons to consider 
for a rectilinear polygon with $n$ vertices in our algorithm.
To obtain an \mip, our algorithm considers a number of uni-rectangle partitions and
chooses the one with the minimum total length of the cuts among the partitions.
From the type classification and tight analyses, 
we show that there are $O(n^3)$ uni-rectangle partitions to consider in total.
The length of the cuts in each uni-rectangle partition can be computed in $O(1)$ time, by maintaining some relevant information.
Our algorithm uses $O(1)$ space for partitioning a subpolygon, which corresponds to a subproblem in our algorithm. It also maintains $O(n)$ arrays of length $O(n)$
to partition subpolygons of a certain type efficiently.
Thus, our algorithm returns an \mip of $P$ in $O(n^3)$ time using $O(n^2)$ space.

For a rectangle containing $m$ point holes, 
Gonzalez and Zheng~\cite{gonzalezA} gave an $O(m^4)$-time 3-approximation algorithm
that transforms the rectangle into a \emph{weakly-simple polygon}\footnote{A polygon is weakly-simple if for every $\epsilon > 0$, its vertices can be perturbed by at most $\epsilon$ to obtain a simple polygon.} with $O(m)$ vertices and 
no point holes in $O(m^2)$ time. Then, it computes an \mip of the polygon in $O(m^4)$ time.
By replacing the \mip algorithm with
our $O(m^3)$-time algorithm, we can 
get a 3-approximation algorithm with $O(m^3)$ time for the problem.

We give a \tp algorithm for $P$ under the vertex incidence,
by modifying the \tp algorithm by O'Rourke and Tewari~\cite{oRourke}.
Our algorithm handles subpolygons efficiently by classifying them and enumerating 
uni-rectangle partitions without duplicates, while guaranteeing 
a \tp under the vertex incidence.
It uses certain coherence among uni-rectangle partitions
and compares them in $O(n^3 \log^2 n)$ time and $O(n^3)$ space in total.
This guarantees the desired time and space complexities of our \tp algorithm under 
the vertex incidence.

Our \tp algorithm under the vertex incidence 
can be used to compute a \tp under the boundary incidence with some modifications.
Under the boundary incidence,
there is an optimal cutset consisting of
cuts, each lying on one of $O(n^2)$ horizontal and vertical lines,
defined by pairs of vertices of $P$. 
Thus, our algorithm takes $O(n^6 \log^2 n)$ time and $O(n^6)$ space to compute a \tp under 
the boundary incidence.

Finally, for a rectilinear polygon $P$ with holes, we consider the decision problem $\myth(P, \delta, k)$
for a positive real value $\delta$ and a positive integer $k$,
determining whether there exists a rectangular partition $\pt$ 
consisting of at most $k$ rectangles with side lengths at least $\delta$ under the vertex incidence.
We show that $\myth$ is NP-hard, using a polynomial-time reduction from the planar 3-satisfiability (P3SAT) problem~\cite{lichtenstein1982planar}, which is known to be NP-complete.


This paper is organized as follows.
Section~\ref{preliminaries} provides terms, definitions, notations and a base lemma 
that are used throughout the paper.
Section~\ref{sec:DP_lingas} gives a review on the previous algorithms
for the \mip problem and the \tp problem.
We present our \mip algorithm in Section~\ref{sec:minink_mine}, 
and our \tp algorithm in Section~\ref{sec:thick_partition}. 
We conclude this paper with a few open problems in Section~\ref{sec:conclusion}.

\section{Preliminaries}
\label{preliminaries}
We denote by $P$ the input rectilinear polygon with $n$ vertices and no holes in the plane. 
We assume that $P$ is axis-aligned and it is given as a sequence of vertices in
counterclockwise order along its boundary. 
For a compact set $X$, we use $\myint{X}$ and $\mybd{X}$ to denote the interior and 
the boundary of $X$, respectively.
For any two points $p$ and $q$ in the plane, we denote by $pq$ the line segment connecting them,
and by $R_{pq}$ the axis-aligned rectangle with opposite corners $p$ and $q$.
We use $x(p)$ and $y(p)$ to denote the $x$-coordinate and the 
$y$-coordinate of a point $p$, respectively.
The grid induced by the lines, each extended from an edge of $P$,
is called the \emph{canonical grid} of $P$ and it is denoted by $\gt$.

\begin{figure}[ht]
    \begin{center}
      \includegraphics[width=.8\textwidth]{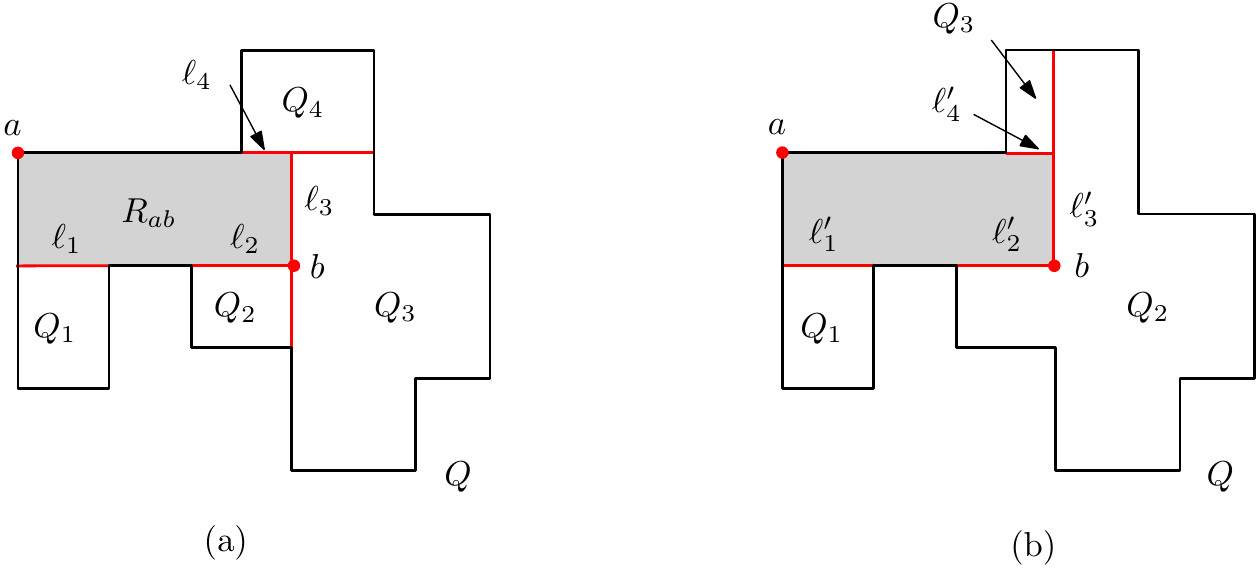}
    \end{center}
    \caption{(a) Partition $\qt = \{R_{ab}, Q_1, Q_2, Q_3, Q_4\}$ of $Q$ and 
    its cutset $\lt = \{\cut_1, \cut_2, \cut_3, \cut_4\}$. 
    $\lt$ determines $R_{ab}$ from $Q$.
    (b) Uni-rectangle partition $\ut = (Q, R_{ab}, \lt')$ with $\lt'=\{\cut'_1, \cut'_2, \cut'_3, \cut'_4\}$.
    $\lt'$ (and thus $\ut$) determines $R_{ab}$ from $Q$.
    }
    \label{fig:definitions}
\end{figure}

\myparagraph{Cutsets.}
A \emph{partition} of a rectilinear polygon is a set of interior-disjoint rectilinear subpolygons 
induced by a \emph{cutset} consisting of open axis-aligned line segments, called cuts, 
contained in the polygon such that the union of the subpolygons 
equals the polygon.
For any partition $\pt$ of a rectilinear polygon,
we denote by $\card{\pt}$ the number of subpolygons in $\pt$. 
We require the cuts in a cutset 
to be \emph{disjoint} and \emph{maximal} with respect to the corresponding partition.
For each partition of a rectilinear polygon $Q$ in Figure~\ref{fig:definitions},
the cuts are disjoint and maximal with respect to the partition.
 
For a subpolygon $Q'$ in a partition $\qt$ of 
a rectilinear polygon $Q$, we say the cutset of $\qt$ 
\emph{determines} $Q'$ from $Q$ if for every cut $\cut$ in the cutset of $\qt$,
$\cut\cap\mybd{Q'}$ is a (nondegenerate) line segment. 
In other words, every cut of the cutset contributes to the boundary of $Q'$.
We also say $\qt$ determines $Q'$ from $Q$ if the cutset of $\qt$ 
determines $Q'$ from $Q$.
Figure~\ref{fig:definitions} shows two cutsets 
that determine rectangle $R_{ab}$ from $Q$. 

\myparagraph{Uni-rectangle partitions.}
A \emph{rectangular partition} of a rectilinear polygon is a partition of the polygon 
into rectangles.
A \emph{uni-rectangle partition} of a rectilinear polygon $Q$ is a partition of $Q$
whose cutset determines a rectangle from $Q$.
Thus, for two points $a,b\in Q$ such that $R_{ab}$ is contained in $Q$,
and a uni-rectangle partition with cutset $\lt$ determining $R_{ab}$ from $Q$, 
we use $(Q, R_{ab}, \lt)$ to denote 
the uni-rectangle partition of $Q$ for $R_{ab}$ induced by $\lt$.
Observe that there can be more than one uni-rectangle partition 
whose cutset determines $R_{ab}$ from $Q$
as there can be more than one such cutset. 
See Figure~\ref{fig:definitions} for two different uni-rectangle partitions determining 
$R_{ab}$ from $Q$.

We use the following notations for the amounts of ink used in a rectangular partition,
in an \mip, and in a uni-rectangle partition of a polygon.
We denote by $\len{\cut}$ the length of a line segment $\cut$.
\begin{itemize}
\item $\myink{\rt} := \sum_{\cut \in \lt} \len{\cut}$, for a rectangular partition $\rt$ and its cutset $\lt$.
\item $\myink{Q}$ is the minimum of $\myink{\rt}$ over all rectangular partitions $\rt$ 
of a rectilinear polygon $Q$.
\item $\myink{Q, R_{ab}, \lt} := \sum_{\cut \in \lt} \len{\cut} + \sum_{Q' \in \ut} \myink{Q'}$, 
for a uni-rectangle partition $\ut = (Q, R_{ab}, \lt)$ of a rectilinear polygon $Q$.
\end{itemize}

For a rectilinear polygon $Q$, a rectangular partition $\rt$ realizing $\myink{Q} = \myink{\rt}$ is called an \mip of $Q$.
A uni-rectangle partition $\ut=(Q, R_{ab}, \lt)$ of $Q$ is \emph{optimal}
if there is an \mip $\rt$ of $Q$ consisting of $R_{ab}$
and the rectangles from some rectangular partitions of the subpolygons of $\ut$.
In this case, $\rt$ is an \emph{optimal refinement} of $\ut$.
Lemma~\ref{lem:uni} 
comes from the fact that each vertex of a rectilinear polygon is used as a corner of a rectangle 
in every rectangular partition of the polygon.
\begin{lemma}
    \label{lem:uni}
For a rectilinear polygon $Q$,
     $\myink{Q} =\min_{b} \min_{\lt} \myink{Q, R_{ab}, \lt}$
     for any fixed vertex $a$ of $Q$ and any point $b \in Q$.
\end{lemma}
\begin{proof}
    Since any optimal refinement of a uni-rectangle partition is 
    a rectangular partition of $P$,
    we have $\myink{Q} \le \min_{b} \min_{\lt} \myink{Q, R_{ab}, \lt}$.
    
    Assume to the contrary that $\myink{Q} < \min_{b} \min_{\lt} \myink{Q, R_{ab}, \lt}$.
    Then there is no \mip $\rt$ of $Q$ such that
    $a$ is a corner of a rectangle in $\rt$. 
    If $a$ is a convex vertex of $Q$, there is a rectangle with a corner on $a$ 
    in every \mip, 
    a contradiction. So $a$ must be a reflex vertex of $Q$.
    By the assumption, $a$ lies in the interior of a side of a rectangle $R$ 
    in some \mip $\rt'$ of $Q$. Then there is a subpolygon in $Q\setminus R$
    such that $a$ is a convex vertex of the subpolygon. 
    This implies that there is a rectangle with a corner at $a$ in every rectangular 
    partition of the subpolygon, and there is a rectangle in $\rt'$ with a corner at $a$, a contradiction.
\end{proof}

\myparagraph{Vertex cuts and anchored cuts.} 
Cuts drawn inside a rectilinear polygon $Q$ can be classified into two types.
A cut is a vertex cut (\vcut, in short) if it is incident to a reflex vertex of $Q$.
A cut is an anchored cut (\acut, in short) if it is incident to the boundary of $Q$.
Note that every \vcut is an \acut.
A partition is said to be under the \emph{vertex incidence} 
if every cut of the partition is a \vcut,
and under the \emph{boundary incidence} if every cut of the partition is an \acut.
Figure~\ref{fig:definitions}(a) shows a cutset consisting of \vcuts, and thus
the partition is under the vertex incidence.
In Figure~\ref{fig:definitions}(b), $\cut'_3$ is an \acut, but not a \vcut. Thus, the partition
is under the boundary incidence, but not under the vertex incidence.

A subpolygon $Q$ of $P$ is a \emph{$k$-cut subpolygon} of $P$ 
if the boundary of $Q$ consists of a contiguous boundary portion of $P$ 
and a chain consisting of $k$ line segments contained in $\myint{P}$, 
alternating horizontal and vertical.
The $k$ line segments constitute the cutset that partitions $P$ into 
$Q$ and $P\setminus Q$. We call each such segment a \emph{boundary cut} of $Q$.
Then $P$ itself is the 0-cut subpolygon of $P$. 
A subpolygon is a $\le$$k$-cut subpolygon if it is a $k'$-cut subpolygon for some $k' \le k$.

\section{Previous algorithms for rectangular partitions}
\label{sec:DP_lingas}
We give a sketch of the algorithm by Lingas et al.~\cite{lingas} for the \mip problem 
and a sketch of the algorithm by O'Rourke and Tewari~\cite{oRourke} for the \tp problem.
Both algorithms employ a dynamic programming approach, and the latter one is 
based on the former one.

\subsection{\mip algorithm by Lingas et al.}
\label{sec:minink_lingas}
Lingas et al. observed that there is an \mip of $P$ 
by a cutset $\lt$ consisting of \acuts lying on grid lines of the canonical grid $\gt$, and 
gave an algorithm returning $\lt$. The rectangles of 
the partition induced by $\lt$ have corners all at grid points of $\gt$.

Their algorithm works recursively as follows.
Let $Q$ denote the input subpolygon for the recursive algorithm. 
Initially, $Q = P$.
The algorithm finds a uni-rectangle partition for a rectangle $R\subset Q$
such that $R$ has its corners at grid points and it is incident to 
a boundary cut of $Q$. In doing so, it fixes a vertex of $Q$ 
as the \emph{origin point} $o$, and searches for every grid point $p$ (called a \emph{partner point} of $o$) of $\gt$ 
such that $R_{op}\subset Q$ and
each subpolygon in a uni-rectangle partition $(Q, R_{op}, \lt)$, except $R_{op}$, 
is either a 1-cut or 2-cut subpolygon of $P$.
This is called the \emph{2-cut property}.
It finds a uni-rectangle partition 
$\ut^{*} = (Q, R_{op^*}, \lt^*)$ satisfying 
$\myink{Q, R_{op^*}, \lt^*} = \min_{p} \min_{\lt} \myink{Q, R_{op}, \lt}$.
Thus, $\myink{Q}=\myink{Q, R_{op^*}, \lt^*}$.
Finally, it computes an optimal refinement $\rt_{\ut^{*}}$ by computing
an \mip 
of each subpolygon in $\ut^{*} = (Q, R_{op^*}, \lt^*)$ 
recursively.
By Lemma~\ref{lem:uni}, $\rt_{\ut^{*}}$ is an \mip of $Q$.
Their dynamic programming algorithm uses $\myink{Q'}$ value stored in a table
for each subpolygon $Q' \in (Q, R_{op}, \lt)$.
Figure~\ref{fig:minink_runs} illustrates how the algorithm works.

The origin point $o$ of $Q$ is any convex vertex of $P$ for $Q=P$ (Figure~\ref{fig:minink_runs}(a)),
the endpoint shared by the two boundary cuts for a 2-cut subpolygon $Q$ (Figure~\ref{fig:minink_runs}(b)),
and an endpoint of the edge containing the boundary cut for a 1-cut subpolygon $Q$ (Figure~\ref{fig:minink_runs}(c)).

\begin{figure}[ht]
    \begin{center}
      \includegraphics[width=.9\textwidth]{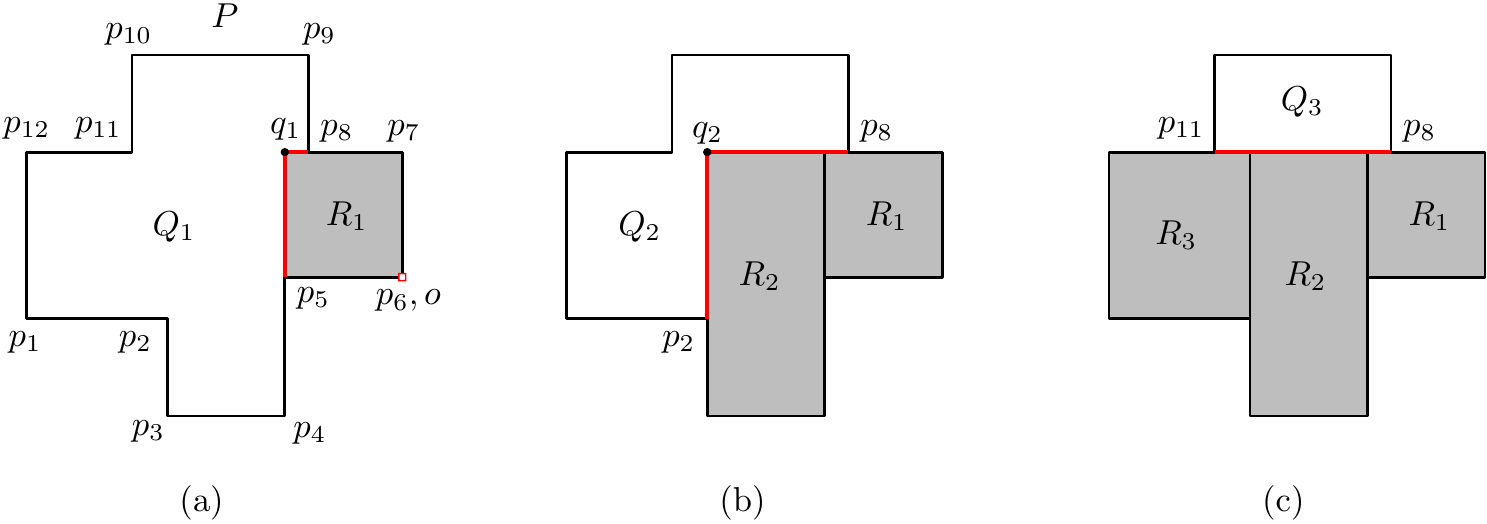}
    \end{center}
    \caption{The \mip algorithm by Lingas et al. in a series of recursive steps.
    During the algorithm, we encounter 2-cut subpolygons $Q_1$ (by cuts $q_1p_8, q_1p_5$) in (a) 
    and $Q_2$ (by cuts $q_2p_8, q_2p_2$) in (b), and 
    a 1-cut subpolygon (rectangle) $Q_3$ (by cut $p_{11}p_8$) in (c).
    The resulting \mip is $\{R_1, R_2, R_3, Q_3\}$ with cutset $\{p_{11}p_8, q_1p_5, q_2p_2\}$.
    }
    \label{fig:minink_runs}
\end{figure}

\begin{figure}[ht]
    \begin{center}
      \includegraphics[width=.9\textwidth]{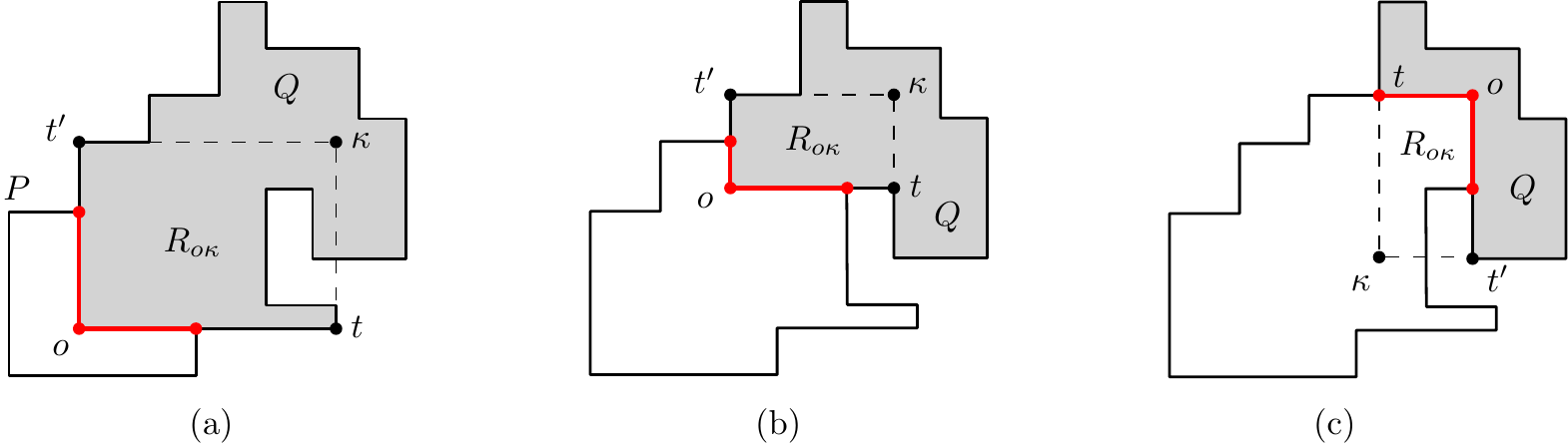}
    \end{center}
    \caption{Polygon $P$ and 2-cut subpolygons $Q$ (gray).
    For a kitty corner $\kc$ of $o$, (a) $\kc\in Q$, but $R_{o\kc}\not\subseteq Q$,
    (b) $\kc\in Q$, and $R_{o\kc}\subseteq Q$.
    (c) $\kc\not\in Q$, and thus $R_{o\kc}\not\subseteq Q$.
    }
    \label{fig:kittycorner}
\end{figure}

Lingas et al. claimed that for a fixed origin point $o$, 
it suffices to check only certain partner points with respect to $R_{o\kc}$ 
to compute an \mip while maintaining the 2-cut property, where
$\kc$ is the \emph{kitty corner} of $o$.
The kitty corner $\kc$ of the origin point $o$ for a 2-cut subpolygon $Q$ is 
a grid point of $\gt$ such that 
$ot$ and $ot'$ are edges of $Q$ for corners $o,t,\kc,t'$ of $R_{o\kc}$. 
Note that $\kc$ is not necessarily a partner point of $o$, and it may lie outside of $Q$. 
See Figure~\ref{fig:kittycorner}.

\begin{figure}[hb]
    \begin{center}
      \includegraphics[width=.9\textwidth]{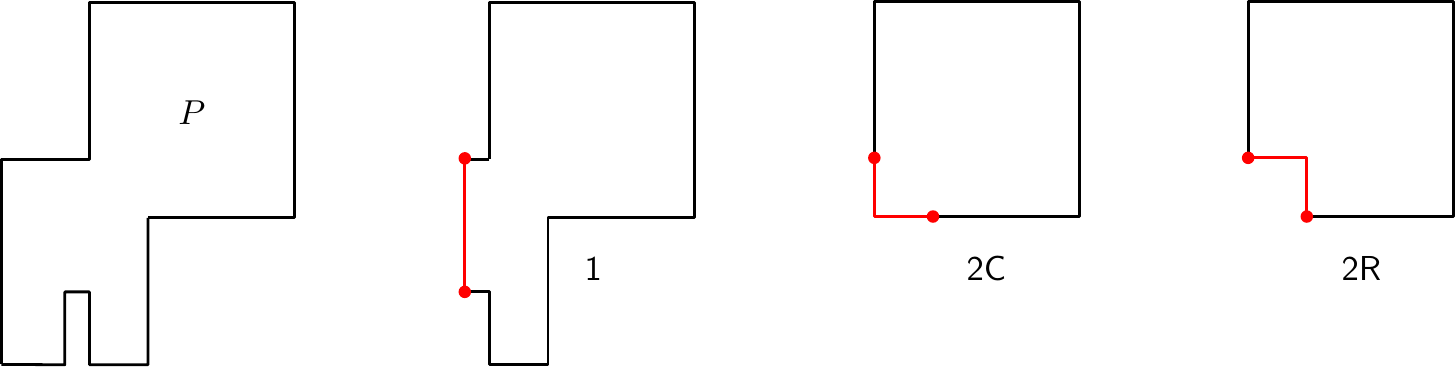}
    \end{center}
    \caption{Three different types of $\le$2-cut subpolygons of $P$, 
    and their boundary cuts (red). 
    }
    \label{fig:three_types}
\end{figure}

Every 1-cut subpolygon is of type \one.
A 2-cut subpolygon $Q$ is classified into two subtypes: the endpoint shared 
by the two cuts is either convex (type \twoc) or reflex (type \twor) with respect to $Q$. See Figure~\ref{fig:three_types}.
The following two rules summarize their claim for a 2-cut subpolygon $Q$ 
with origin point $o$, kitty corner $\kc$, and a partner point $p$ of $o$.

\begin{itemize}
    \item[-] For $Q$ of type \twoc,
    $R_{op}$ appears in an \mip of $Q$ only if 
    $p\not\in\myint{R_{o\kc}}$.
    Every subpolygon in uni-rectangle partition $(Q, R_{op}, \lt)$ is a $\le$2-cut subpolygon, 
    where $\lt$ is any cutset consisting of \acuts and determining $R_{op}$.

    \item[-] For $Q$ of type \twor, $R_{op}$ appears in an \mip of $Q$ only if
    either the vertical line or the horizontal line through $p$ intersects $R_{o\kc}$.
    If so, every subpolygon in uni-rectangle partition $(Q, R_{op}, \lt)$ is a $\le$2-cut subpolygon,
    where $\lt$ is any cutset consisting of \acuts and determining $R_{op}$.
\end{itemize}

Any partner point satisfying one of the two rules above is a \emph{candidate partner point} 
in their algorithm. However, there are some flaws in the rules, which can be fixed easily. 
See Section~\ref{sec:minink_mine} for details.
They claimed that any 1-cut subpolygon can be handled in the same way for 
a 2-cut subpolygon by considering a subpolygon edge, incident and perpendicular 
to the boundary cut, as an additional boundary cut. 
However, as O'Rourke and Tewari~\cite{oRourke} pointed out, 
this is not always the case. Again, this can be corrected by choosing
the origin and partner points in a certain way. 
See Section~\ref{sec:minink_mine} for details.

Their algorithm compares $\myink{\ut}$ values of uni-rectangle partitions $\ut$ with 
$R_{op}$ satisfying the rules above, and takes the uni-rectangle partition 
with the minimum $\myink{\ut}$ value.
It returns an \mip of $P$ in $O(n^4)$ time using $O(n^2)$ space, 
because there are $O(n^2)$ $\le$2-cut subpolygons,
and $O(n^2)$ candidate partner points for each subpolygon.

\subsection{\tp algorithm by O'Rourke and Tewari}
\label{sec:thick_oRourke}
We use \emph{\vtp} to refer to a \tp whose cutset consists of \vcuts.
We give a sketch of the algorithm by O'Rourke and Tewari~\cite{oRourke}, denoted by \tpv,
that computes a \vtp. 

\tpv makes use of observations similar to the ones in the \mip algorithm
by Lingas et al.~\cite{lingas} under the vertex incidence.
In the following, we use terms in Section~\ref{sec:minink_lingas}
such as the 2-cut property, origin points, and partner points.
The algorithm maintains the 2-cut property 
and handles subpolygons belonging to each of the three types (\one, \twoc, \twor) 
separately, as the \mip algorithm by Lingas et al. does.
For subpolygons of types \twoc and \twor, it uses a few rules for determining origin points 
and their partner points in order to compute a \vtp.
A subpolygon of type \one is considered as a subpolygon of type \twoc or \twor.

However, there are degenerate 1-cut subpolygons $Q$ such that
the two edges of $P$ connected by the boundary cut $\cut$ of $Q$ 
are collinear to $\cut$.
To handle such a degenerate 1-cut subpolygon,
\tpv takes each of the $O(n)$ grid points of $\gt$ lying on the grid line through 
its boundary cut $\cut$ and contained in $Q$ as an origin point. 
For each origin point $o$, it finds partner points $p$ such that
$\cut$ intersects the boundary of $R_{op}$ in a (nondegenerate) line segment.
Observe that every subpolygon in the uni-rectangle partition 
is again a $\le$2-cut subpolygon.
See Figure~\ref{fig:degenerate}. 
Hence, \tpv obtains a \vtp of each 1-cut subpolygon 
by checking all such origin and partner point pairs.

Based on the type classification and analyses, together with dynamic programming,
\tpv computes a \vtp of $P$
in $O(n^5)$ time and $O(n^4)$ space.

\begin{figure}[ht]
    \begin{center}
      \includegraphics[width=.4\textwidth]{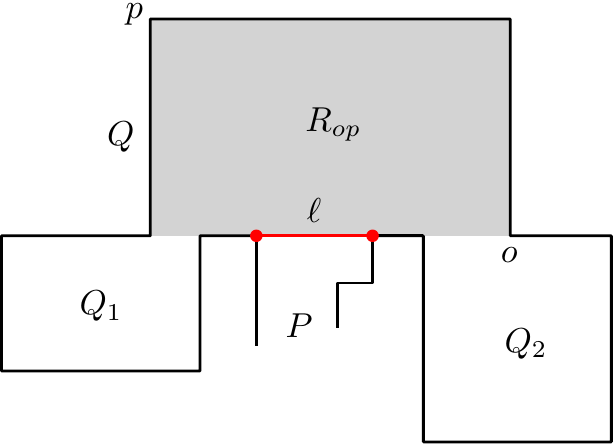}
    \end{center}
    \caption{A degenerate 1-cut subpolygon $Q$ whose boundary cut $\cut$ 
    connects two edges of $P$ collinear to $\cut$.
    Since $\cut$ intersects the boundary of $R_{op}$ in a line segment, 
    $p$ is considered as a partner point of $o$ by \tpv. 
    The resulting partition $\{R_{op}, Q_1, Q_2\}$ is the \vtp of $Q$.
    }
    \label{fig:degenerate}
\end{figure}

\section{\mip algorithm}
\label{sec:minink_mine}
We present an efficient algorithm that computes 
an \mip of a rectilinear polygon $P$ with $n$ vertices and no holes
in the plane.
Our algorithm is based on the work by Lingas et al.~\cite{lingas},
but with careful analyses on the subproblems and by exploiting certain
geometric and combinatorial coherence among them, we could improve upon their algorithm.

We classify subpolygons into four types (\onec, \oner, \twoc, \twor) 
and enumerate uni-rectangle partitions of each type
without duplicates, while guaranteeing an \mip.
We show that there are $O(n^3)$ uni-rectangle partitions of 
types \onec, \oner, and \twoc, and $O(n^4)$ uni-rectangle partitions of type \twor, 
to consider for an \mip.
We use certain coherence among uni-rectangle partitions of type \twor 
so that it suffices to check only $O(n^3)$ of them.
We can compute $\myink{\cdot}$ value for each uni-rectangle partition in $O(1)$ time,
after $O(n^2)$ time preprocessing (details in Section~\ref{sec:made_clear}),
so our algorithm runs in $O(n^3)$ time.

We show that our algorithm uses $O(n^2)$ space. 
The preprocessing step for $\myink{\cdot}$ queries and
the query data structure maintained in the algorithm
requires $O(n^2)$ space. 
The algorithm also maintains $O(n)$ arrays of length $O(n)$ to efficiently compute 
$\myink{\cdot}$ values of uni-rectangle partitions of type \twor.
Thus, our algorithm uses $O(n^2)$ space in total.

Our algorithm is based on the following lemma that guarantees an \mip 
by a cutset consisting of \vcuts. 
This is stronger than the observation by Lingas et al. using \acuts in Section~\ref{sec:minink_lingas}.

\begin{lemma}[Lemma 5.1 of \cite{duD}]
    \label{lem:vertex_cut_minink}
    There is an \mip of $P$ by a cutset consisting of \vcuts.
\end{lemma}
By Lemma~\ref{lem:vertex_cut_minink}, 
our algorithm considers cutsets consisting of \vcuts only.
It maintains the 2-cut property and uses the type classification of 
2-cut subpolygons into \twoc and \twor in 
Section~\ref{sec:minink_lingas}.
We classify 1-cut subpolygons $Q$ further into \oner 
(degenerate 1-cut subpolygons mentioned in Section~\ref{sec:thick_oRourke})
and \onec (1-cut subpolygons not belonging to \oner).

Throughout this section, we denote by $Q$ a $\le$2-cut subpolygon.
For any 2-cut subpolygon $Q$, we use $\cut_h$ (horizontal) and $\cut_v$ (vertical) 
to denote its two boundary cuts.

\myparagraph{Choosing origin points.}
If $Q$ is of type \twoc or \twor, the endpoint shared by two cuts is the origin point $o$.
We use $\kc$ to denote the kitty corner of $o$.
If $Q$ is of type \onec, an endpoint of its boundary cut lying on a convex vertex of $Q$
is the origin point $o$.
If $Q$ is of type \oner, the grid points contained in $Q$
and lying on the grid line through the boundary cut of $Q$ are the origin points. 
Thus, there are $O(n)$ origin points to consider for $Q$
of type \oner.

\myparagraph{Optimal uni-rectangle partitions.}
In the following lemma, we characterize the uni-rectangle partitions that 
appear in an \mip with respect to the subpolygon type (\twoc, \twor, and \oner),
and give rules for choosing partner points based on the characterization.
Each subpolygon $Q$ of type \onec can be considered 
as a subpolygon of type \twoc,
by using an edge of $Q$ incident and perpendicular to its boundary cut, 
as an additional boundary cut.
Then, rule 1 in the following lemma can be applied to the induced subpolygon of type \twoc.

\begin{lemma}
\label{lem:minink_rules}
    Given a $\le$2-cut subpolygon $Q$ of $P$ and its origin point $o$, 
    we can compute an optimal uni-rectangle partition of $Q$ while maintaining the 2-cut property,
    by checking partner points $p$ satisfying one of the following three rules.
    \begin{itemize}
    \item[1.] For $Q$ of type \twoc or \onec, 
    (1) a side of $R_{op}$ incident to $p$ intersects $\mybd{Q}$ in a (nondegenerate) line segment
    or (2) $p \notin \myint{R_{o\kc}}$.
    \item[2.] For $Q$ of type \twor,
    (1) a side of $R_{op}$ incident to $p$ intersects $\mybd{Q}$ in a (nondegenerate) line segment
    or (2) $p$ lies in a quadrant (defined by the horizontal line and the vertical line through $o$) not containing $\kc$ and not opposite to the quadrant containing $\kc$. 
    \item[3.] For $Q$ of type \oner, the boundary cut of $Q$ intersects 
    the boundary of $R_{op}$ in a (nondegenerate) line segment.
    \end{itemize}
\end{lemma}
\begin{proof} 
Lingas et al.~\cite{lingas} claim (in \emph{Lemma 3} of the paper) that 
for a subpolygon $Q$ of type \twoc and its origin point $o$, 
$p \notin \myint{R_{o\kc}}$ for every partner point $p$.
The claim holds only when no side of $R_{op}$ incident to $p$ 
intersects $\mybd{Q}$ in a nondegenerate line segment. Thus, we have rule 1.
They also claim (in \emph{Lemma 4} of the paper) that
for each subpolygon $Q$ of type \twor and its origin point $o$, 
either the vertical line or the horizontal line through $p$ intersects $R_{o\kc}$,
for the kitty corner $\kc$ of $o$.
Thus we can rule out the partner points $p$ in the quadrant 
(defined by the horizontal line and the vertical line through $o$) 
opposite to the quadrant containing $\kc$
such that both sides of $R_{op}$ incident to $p$ do not intersect $\mybd{Q}$ in a nondegenerate line segment.
    Finally, we can show that the iteration over the partner points satisfying rule 3
    can be done by checking all uni-rectangle partitions of each subpolygon of type \oner
    without duplicates, using an argument similar to the proof of Lemma~\ref{lem:uni}.
\end{proof}

\subsection{Counting distinct \texorpdfstring{$(Q, o, p)$}{(Q,o,p)} triplets}
\label{sec:subpolygons}
Our algorithm computes $\myink{\cdot}$ values for the uni-rectangle partitions 
satisfying the 2-cut property and Lemma~\ref{lem:minink_rules}.
It takes time linear to the number of such uni-rectangle partitions
because we can compute $\myink{\cdot}$ value for each uni-rectangle partition
in $O(1)$ time, after preprocessing. 

To bound the number of such uni-rectangle partitions, it suffices to count 
the distinct $(Q, o, p)$ triplets (called \emph{candidate triplets})
in the uni-rectangle partitions, because for a triplet $(Q, o, p)$, 
there are $O(1)$ distinct cutsets $\lt$ consisting of \vcuts 
for uni-rectangle partitions $(Q, R_{op}, \lt)$.
We will show this in Section~\ref{sec:made_clear}.

We classify candidate triplets $(Q, o, p)$ into types by the types of $Q$,
and count the distinct candidate triplets in the following order of types, 
\twoc, \onec, \oner, \twor, for ease of description.
For any type $\textsf{A} \in \{\twoc, \onec, \oner, \twor\}$, 
$\candset{\textsf{A}}$ denotes the number of distinct candidate triplets of type $\textsf{A}$. 

For a 2-cut subpolygon $Q$, the \emph{cut direction} of $Q$ is 
\emph{up-right} if one boundary cut lies above $o$ and the other boundary cut lies to the right of $o$.  
The \emph{vertex type} of $R_{o\kc}$ for $Q$ with cut direction up-right
is \textit{convex-reflex} if the bottom-right corner of $R_{o\kc}$ is a convex vertex of $Q$, 
and the top-left corner of  $R_{o\kc}$ is a reflex vertex of $Q$.
Because $o$ and $\kc$ are fixed for $Q$, we say that the vertex type of $Q$ is 
convex-reflex if the vertex type of $R_{o\kc}$ is convex-reflex.




\subsubsection{Candidate triplets of type \twoc.}
\label{subsec:2C}
We count the candidate triplets for $Q$ of type \twoc.
We first show that there are $O(n^3)$ candidate triplets for $Q$ 
satisfying condition (1) of 
rule 1 in Lemma~\ref{lem:minink_rules}.
Then, we consider the candidate triplets for $Q$ satisfying (2) but not satisfying (1) of 
rule 1 in Lemma~\ref{lem:minink_rules},
classify them into three subtypes, 
and then bound the number of candidate triplets of each subtype to $O(n^3)$.

We first count the candidate triplets for $Q$ of type \twoc that satisfy condition (1).
The grid points $p$ such that $(Q, o, p)$ satisfies condition (1) of rule 1 in Lemma~\ref{lem:minink_rules}
are on a staircase chain of line segments on grid lines of $\gt$ which 
alternates between horizontal and vertical.
See Figure~\ref{fig:2C_subpolygons}(a).
Since there are $O(n^2)$ subpolygons of type \twoc and  
$O(n)$ grid points of $\gt$ lie on such a chain, 
there are $O(n^3)$ candidate triplets of type \twoc satisfying condition (1).

We then count the candidate triplets for $Q$ of type \twoc that 
do not satisfy condition (1) but do satisfy condition (2) of 
rule 1 in Lemma~\ref{lem:minink_rules}.
To do this efficiently, 
we classify candidate triplets further into three subtypes by 
the vertex types of the corners $t, t'$ of $R_{o\kc}$ other than $o$ and $\kc$.
If both $t$ and $t'$ are convex vertices of $Q$,
$Q$ is \emph{closed} (subtype \twocc).
If both $t$ and $t'$ are reflex vertices of $Q$,
$Q$ is \emph{open} (subtype \twoco).
If one is a convex vertex and the other is a reflex vertex of $Q$, 
$Q$ is \emph{half-open} (subtype \twoch).
We assume that the cut direction of $Q$ is up-right 
with $\cut_h$ and $\cut_v$. 

\begin{figure}[ht]
    \begin{center}
      \includegraphics[width=\textwidth]{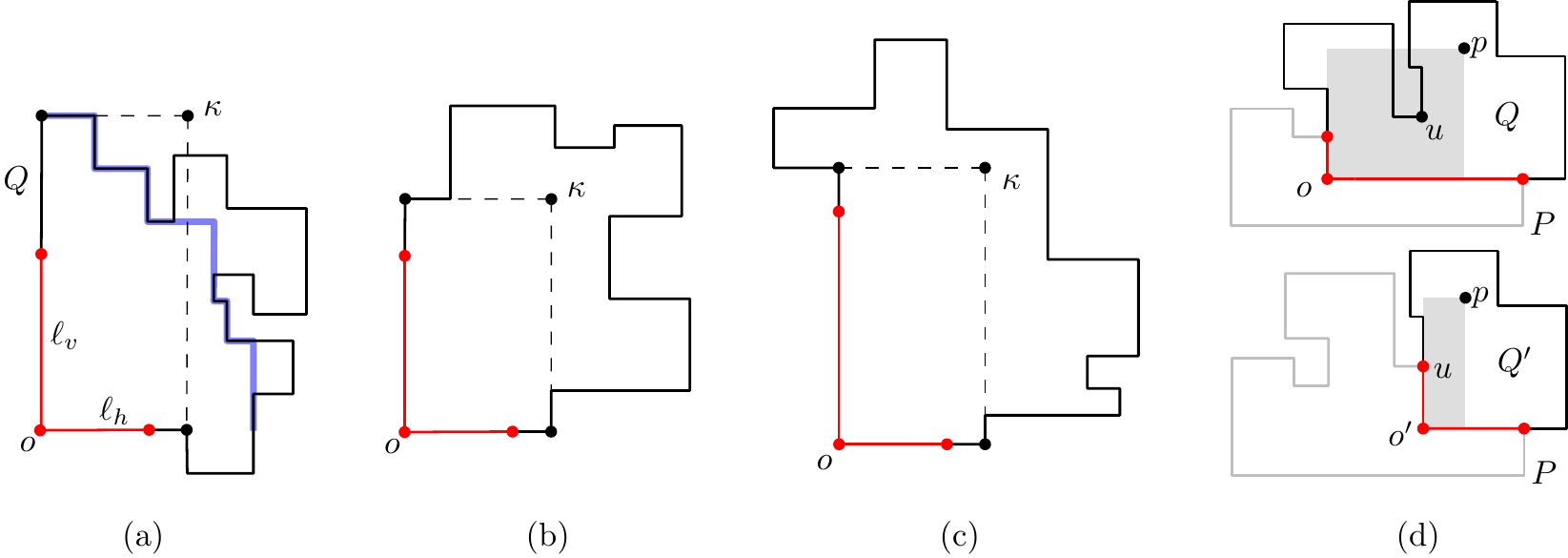}
    \end{center}
    \caption{(a) For $Q$ of type \twoc, 
    the grid points $p$ such that $(Q, o, p)$ satisfies condition (1) of 
    rule 1 in Lemma~\ref{lem:minink_rules}
    are on a staircase chain of line segments on grid lines, alternating between horizontal and vertical.
    (b) A subpolygon of type \twocc.
    (c) A subpolygon of type \twoch. 
    (d) Subpolygons $Q$ and $Q'$ of type \twoch with
    origin points lying on the same horizontal grid line.
    }
    \label{fig:2C_subpolygons}
\end{figure}

\myparagraph{When $Q$ is closed (\twocc).} 
Observe that $R_{op}\subseteq Q$ 
only if $p \in R_{o\kc}$. See Figure~\ref{fig:2C_subpolygons}(b). 
On the other hand, 
$p \notin \myint{R_{o\kc}}$ by condition (2) of 
rule 1 in Lemma~\ref{lem:minink_rules}.
Hence, if $(Q, o, p)$ is a candidate triplet, $p$ lies on the boundary of $R_{o\kc}$.
More precisely, $p$ lies on one of the two sides of $R_{o\kc}$ incident to $\kc$, since otherwise $R_{op}$ becomes a line segment. 
Because each such $p$ satisfies condition (1) of rule 1 in Lemma~\ref{lem:minink_rules},
there is no candidate triplet of type \twocc.

\myparagraph{When $Q$ is half-open (\twoch).} 
Observe that $R_{op}\subseteq Q$ only if $p$ lies in the intersection of $Q$ and 
the vertical slab bounded by the two lines, one through the left side of $R_{o\kc}$
and one through the right side of $R_{o\kc}$.
See Figure~\ref{fig:2C_subpolygons}(c).
On the other hand, 
$p \notin \myint{R_{o\kc}}$ by condition (2) of 
rule 1 in Lemma~\ref{lem:minink_rules}.
Hence, if $(Q, o, p)$ is a candidate triplet, $p$ lies in the intersection of $Q\setminus \myint{R_{o\kc}}$
and the vertical slab.

We need the following lemma to bound the number of 
candidate triplets of type \twoch.

\begin{lemma}
    \label{lem:2Ch_bound}
    Let $Q$ and $Q'$ be subpolygons of type \twoch 
    with cut direction up-right and vertex type convex-reflex.
    If the origin points $o$ of $Q$ and $o'$ of $Q'$ lie 
    on the same horizontal grid line,
    at most one of $(Q, o, p)$ and $(Q', o', p)$ is a candidate triplet 
    for any fixed grid point $p$.
\end{lemma}
\begin{proof}
Let $Q$ and $Q'$ be such subpolygons 
of type \twoch with origin points $o$ and $o'$, respectively, 
lying on the same horizontal grid line with $x(o) < x(o')$.
Since $x(o) < x(o')$, either the horizontal cut of $Q'$ is contained in the horizontal cut of $Q$ (see Figure~\ref{fig:2C_subpolygons}(d)) or 
there is an edge of $P$ contained in the horizontal line segment with endpoints $o$ and $o'$.

Let $u$ be the upper endpoint of the vertical cut of $Q'$.
If both $(Q, o, p)$ and $(Q', o', p)$ are candidate triplets, then 
$x(o) < x(o') = x(u) < x(p)$ and $y(o) = y(o') < y(u) < y(p)$, 
because $p$ satisfies condition (2) of
rule 1 in Lemma~\ref{lem:minink_rules}.
Thus, $u\in\myint{R_{op}}$ and $R_{op} \not\subset Q$ 
because $u$ is a reflex vertex of $Q$.
This contradicts that $(Q, o, p)$ is a candidate triplet.
\end{proof}

There are $O(n^2)$ grid points of $\gt$. 
By Lemma~\ref{lem:2Ch_bound},
each grid point $p$ induces at most one candidate triplet $(Q, o, p)$
among subpolygons $Q$ of type \twoch whose cut direction is up-right, 
whose vertex type is convex-reflex,
and whose origin points lie on the same horizontal grid line of $\gt$.
So there are $O(n^2)$ candidate triplets for such subpolygons of type \twoch 
with the same cut direction,
the same grid line of $\gt$ on which origin points lie, and the same vertex type 
which is determined by the cut direction and the grid line.
There are $O(1)$ cut directions and vertex types, 
and $O(n)$ grid lines of $\gt$. 
Hence, $\candset{\twoch}=O(n^3)$.

\myparagraph{When \texorpdfstring{$Q$}{Q} is open (\twoco).}
\label{sec:2Co_details}
The four regions of the plane subdivided by the vertical line and the horizontal
line through $\kc$ are called the \emph{$\kc$-quadrants} labeled from 
$\kcq{1}$ (top-right region) to $\kcq{4}$
(bottom-right region), in counterclockwise direction around $\kc$.
Since the cut direction of $Q$ is up-right, $o$ always lies in $\kcq{3}$.
We consider the partner points of $o$ 
lying in $\kcq{1}, \kcq{2},$ and $\kcq{4}$,
but no one lying in the interior of $\kcq{3}$, 
by condition (2) of rule 1 in Lemma~\ref{lem:minink_rules}.

We need the following lemma to bound the number of 
candidate triplets for $Q$ of type \twoco.

\begin{lemma}
    \label{lem:2Co_bound}
    Let $Q$ and $Q'$ be subpolygons of type \twoco 
    with up-right cut direction and origin points $o$ and $o'$,
    respectively.
    If $o$ and $o'$ lie on
    the same horizontal grid line of $\gt$, 
    at most one of $(Q, o, p)$ and $(Q', o', p)$ is a 
    candidate triplet for any fixed grid point $p$ in 
    $(\kcq{1}\cup\kcq{2}) \cap (\kcq{1}' \cup \kcq{2}')$,
    for kitty corners $\kc$ of $o$ and $\kc'$ of $o'$.
\end{lemma}
\begin{proof}
    Let $u$ be the upper endpoint of the vertical boundary cut of $Q'$. 
    If $(Q', o', p)$ is a candidate triplet for a grid point 
    $p\in \kcq{1}'\cup\kcq{2}'$
    and $(Q, o, p)$ is a candidate triplet for the same grid point 
    $p\in\kcq{1}\cup\kcq{2}$,
    then $x(o) < x(o') = x(u) < x(p)$ and $y(o) = y(o') < y(u) < y(p)$,
    as in the proof of Lemma~\ref{lem:2Ch_bound}.
    Thus, $u\in\myint{R_{op}}$, and $R_{op} \not\subset Q$ 
    because $u$ is a reflex vertex of $Q$.
    This contradicts that $(Q, o, p)$ is a candidate triplet.
\end{proof}

Similarly, we can prove the following corollary.

\begin{corollary}
    \label{corollary:2Co_bound}
    Let $Q$ and $Q'$ be subpolygons  of type \twoco 
    with up-right cut direction.
    If their origin points $o$ of $Q$ and $o'$ of $Q'$ lie on 
    the same vertical grid line of $\gt$, 
    then at most one of $(Q, o, p)$ and $(Q', o', p)$ is a 
    candidate triplet for any fixed grid point $p$ in 
    $(\kcq{1}\cup\kcq{4}) \cap (\kcq{1}'\cup\kcq{4}')$,
    for kitty corners $\kc$ of $o$ and $\kc'$ of $o'$.
\end{corollary}
By Lemma~\ref{lem:2Co_bound}, 
each grid point $p$ 
induces at most one candidate triplet with $p\in \kcq{1}\cup\kcq{2}$ 
among the subpolygons $Q$ of type \twoco 
with origin points lying on the same
horizontal grid line.
By Corollary~\ref{corollary:2Co_bound},
each grid point $p$
induces at most one candidate triplet with $p\in \kcq{1}\cup\kcq{4}$ 
among the subpolygons $Q$ of type \twoco 
with origin points lying on the same vertical 
grid line.
Since there are $O(n^2)$ grid points and $O(n)$ 
grid lines of $\gt$,
there are $O(n^3)$ candidate triplets $(Q, o, p)$ of type \twoco 
such that the cut direction of $Q$ is up-right.
There are $O(1)$ cut directions, so $\candset{\twoco}=O(n^3)$.

Thus, we bound the number of candidate triplets of type \twoc.
\begin{lemma}
    \label{lem:2C_result}
    There are $O(n^3)$ candidate triplets of type \twoc.
\end{lemma}

\subsubsection{Candidate triplets of type \onec.}
\label{subsec:1C}
Recall that the boundary cut $\cut$ of $Q$ of type \onec, 
together with an additional cut (an edge of $Q$ incident to the edge containing $\cut$
and sharing an endpoint with $\cut$)
induces a \twoc subpolygon.
Since we consider an \mip induced by a cutset consisting of \vcuts,
we have the following lemma.

\begin{lemma}
    There are $O(n)$ subpolygons of type \onec.
    \label{lem:1C_number}
\end{lemma}
\begin{proof}
    One endpoint of the boundary cut of $Q$ of type \onec lies on a reflex vertex of $P$.
    Each reflex vertex of $P$ can be used as an endpoint of the boundary cut 
    for at most two different subpolygons of type \onec.
    There are $O(n)$ reflex vertices of $P$, and thus there are $O(n)$ subpolygons of type \onec.
\end{proof}

For each subpolygon $Q$ of type \onec and its origin point $o$, 
there are $O(n^2)$ partner points $p$ such that $(Q, o, p)$ satisfies 
rule 1 in Lemma~\ref{lem:minink_rules}, 
because there are $O(n^2)$ grid points of $\gt$.
This, together with Lemma~\ref{lem:1C_number}, bounds the number of candidate triplets of type \onec.

\begin{lemma}
    There are $O(n^3)$ candidate triplets of type \onec.
    \label{lem:1C_result}
\end{lemma}

\subsubsection{Candidate triplets of type \oner.}
\label{subsec:1R}
We characterize each subpolygon of type \oner by its cut direction and 
the local placement of the subpolygon around the boundary cut.
Let $H^+$ be the set of the subpolygons $Q$ of type \oner 
such that the boundary cut $\cut$ of $Q$ 
is horizontal and $Q$ lies above $\cut$ locally.
\begin{lemma}
    For $Q$ in $H^+$, 
    $R_{op}$ appears in no \mip of $Q$ 
    if both top corners of $R_{op}$ have their vertical projections onto 
    $\myint{\cut}$.
    \label{lem:1R_vtp}
\end{lemma}
\begin{proof}
    If both top corners of $R_{op}$ have their vertical projections onto $\myint{\cut}$, 
    one of the vertical sides of $R_{op}$ is not on a \vcut, 
    so there is no uni-rectangle partition induced by the triplet $(Q, o, p)$.
    See Figure~\ref{fig:1R_subpolygons}(a).
\end{proof}
\begin{figure}[ht]
    \begin{center}
      \includegraphics[width=\textwidth]{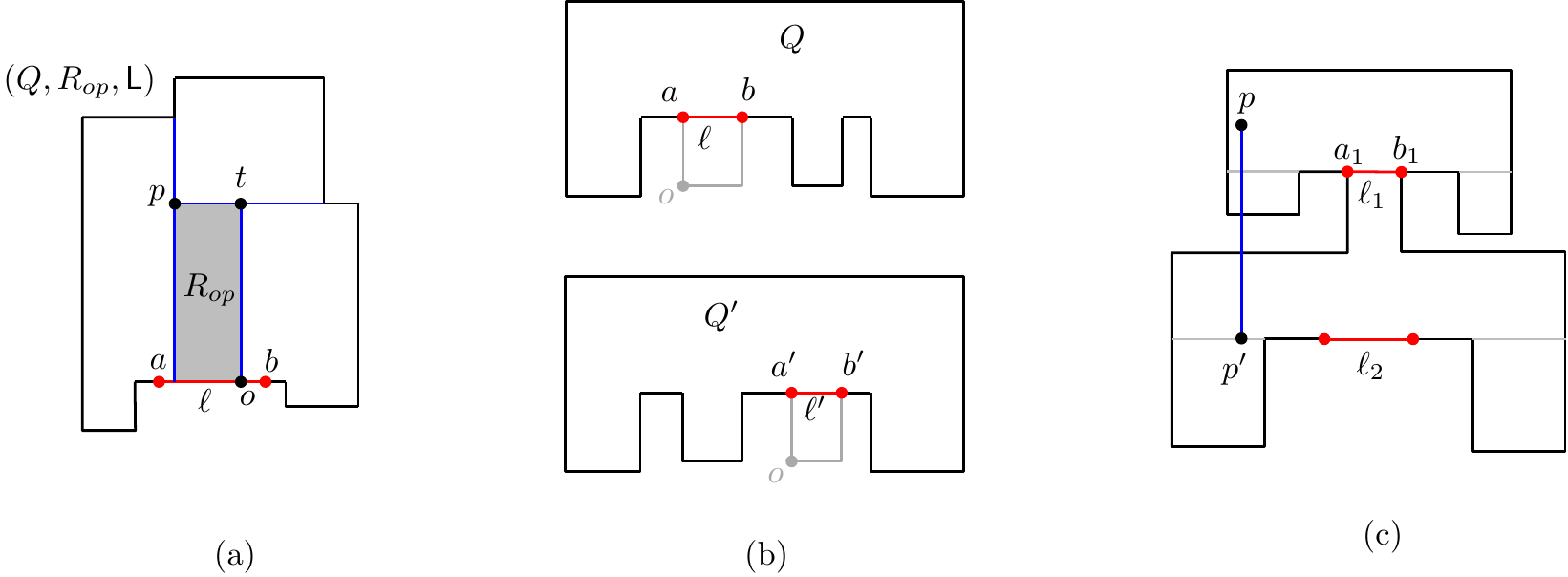}
    \end{center}
    \caption{(a) 
    If both top corners $p$, $t$ of $R_{op}$ have their vertical projections onto $\myint{\cut}$, one vertical side ($ot$ in the figure) of $R_{op}$ is not on a \vcut in $(Q, R_{op}, \lt)$. 
    (b) 
If both $Q$ and $Q'$ are considered, $o \in P(a, b)$ and $o \in P(a', b')$.
This is not possible because $P(a, b) \cap P(a', b') = \emptyset$.
    (c) Proof of Lemma~\ref{lem:1R.reachable}.
    }
    \label{fig:1R_subpolygons}
\end{figure}
\begin{lemma}
    For each grid point $p \in \gt$, the algorithm considers 
    at most one subpolygon $Q$ in $H^+$
    such that $pp^o\subset Q$ and $p^o\not\in\cut$ but $p^oc\subset Q$ for an endpoint $c$
    of $\cut$ and the vertical projection $p^o$ of $p$ onto the line through $\cut$. 
    \label{lem:1R.reachable}
\end{lemma}
\begin{proof}
    For any subpolygon of $P$ (including $P$) and two points $a$ and $b$ on $\mybd{P}$, we denote by $P(a, b)$ 
    the subchain on $\mybd{P}$ in the clockwise direction from $a$ to $b$.

    We first show that among all subpolygons in $H^+$ and 
    boundary cuts lying on the same horizontal grid line of $\gt$, 
    at most one of them is considered by the algorithm. 
    Suppose there are two such subpolygons $Q$ and $Q'$ with (horizontal) boundary cuts $\cut=ab$ and
    $\cut'=a'b'$ considered by the algorithm, respectively.
    Then no vertical line intersects both $\cut$ and $\cut'$.
Moreover, $o \in P(a, b)$ and $o \in P(a', b')$
for the origin point $o$ used in the first recursive step of the algorithm.
This is not possible because $P(a, b) \cap P(a', b') = \emptyset$.
    Thus, there is at most one subpolygon considered by the algorithm
    among subpolygons in $H^+$, 
    with boundary cuts lying on the same horizontal grid line. See Figure~\ref{fig:1R_subpolygons}(b).

    Suppose there exist two subpolygons $Q_1$ and $Q_2$ satisfying the lemma statement, 
with horizontal boundary cuts $\ell_1 = a_1b_1$ and $\ell_2$ of $Q_1$ and $Q_2$, respectively, such that $\ell_1$ lies above $\ell_2$. 
    Then $pp'$ intersects the boundary of $Q_2$ at a point in $Q_2(a_1, b_1)$, contradicting $pp' \subset Q_2$ for  
    the vertical projection $p'$ of $p$ onto the line through $\cut_2$.    
    See Figure~\ref{fig:1R_subpolygons}(c).
\end{proof}

With an argument similar to the one in the proof of Lemma~\ref{lem:1C_number}, 
we can bound the number of subpolygons of type \oner. 

\begin{lemma}
    There are $O(n)$ subpolygons of type \oner.
    \label{lem:1R_number}
\end{lemma}

With Lemmas~\ref{lem:1R_vtp},~\ref{lem:1R.reachable}, and~\ref{lem:1R_number},
we can bound the number of candidate triplets of type \oner.

\begin{lemma}
    There are $O(n^3)$ candidate triplets of type \oner.
    \label{lem:1R_result}
\end{lemma}
\begin{proof}
    We consider only the subpolygons in $H^+$.
    The other cases can be handled similarly.
    By Lemma~\ref{lem:1R_vtp}, 
    we count the candidate triplets $(Q, o, p)$ of type \oner such that
    at least one of the top corners of $R_{op}$, denoted $s$, satisfies 
    $ss^o \subset Q$, $s^o \notin \myint{\cut}$, and $s^oa \subset Q$ for an endpoint $a$ of $\cut$,
    where $s^o$ is the vertical projection of $s$ onto the line through $\cut$.

    First, we count the candidate triplets such that $s^o \notin \cut$.
    By Lemma~\ref{lem:1R.reachable}, there is at most one subpolygon of type \oner  
    for $s$. 
    By checking all $O(n)$ origin points for the subpolygon, 
    there are $O(n)$ candidate triplets for each grid point of $\gt$. 
    Since there are $O(n^2)$ grid points of $\gt$, there are $O(n^3)$ such candidate triplets.

    Then, we count the candidate triplets $(Q, o, p)$ 
    such that $s^o$ lies on an endpoint of $\ell$.
    For each $Q$ of type \oner, there are $O(n^2)$ such candidate triplets 
    because there are $O(n)$ horizontal grid lines of $\gt$ for the choice of $s$ and $O(n)$ vertical grid lines of $\gt$ for the choice of $o$. 
    By Lemma~\ref{lem:1R_number}, 
    there are $O(n^3)$ such candidate triplets in total.

    So, the total number of candidate triplets of type \oner is $O(n^3)$.
\end{proof}

\subsubsection{Candidate triplets of type \twor.}
\label{subsec:2R}
We bound the number of candidate triplets satisfying condition (1) of 
rule 2 in Lemma~\ref{lem:minink_rules} to $O(n^3)$
in a way similar to the one in Section~\ref{subsec:2C}.
For the candidate triplets satisfying condition (2), 
we classify them into two subtypes depending on the position of $p$,
and then bound the number of candidate triplets to be checked by the algorithm
to $O(n^3)$ for each type.

We assume that the cut direction of $Q$ is up-right. 
The four regions of the plane subdivided by the vertical line and the horizontal
line through $o$ are called the \emph{$o$-quadrants}, labeled from 
$\oq{1}$ (top-right region) to $\oq{4}$
(bottom-right region), in counterclockwise direction around $o$.
By condition (2) of 
rule 2 in Lemma~\ref{lem:minink_rules},
the partner point $p$ from 
a candidate triplet $(Q, o, p)$ of type \twor lies either in 
$\oq{2}$ or $\oq{4}$.

Given a candidate triplet $(Q, o, p)$ of type \twor, 
$p$ is \emph{cut-visible} in $Q$ ($(Q, o, p)$ of type \tworv)
if $p$ has an orthogonal projection on $\cut_v$ or on $\cut_h$.
It is \emph{cut-invisible} ($(Q, o, p)$ of type \twori) otherwise.

\myparagraph{Candidate Triplets $(Q, o, p)$ with $p$ cut-invisible (\twori).} 
For each partner point $p$,
there is at most one subpolygon of type \twor per grid line, among those with the same cut direction and cut-invisible $p$.
\begin{lemma}
    \label{lem:origin_second}
    For each grid point $p$, there is at most one subpolygon $Q$ of type \twor
    with origin point $o$
    such that $p \in \oq{4}$ for the candidate triplet $(Q,o,p)$
    among the subpolygons with up-right cut direction,
    origin points lying on the same vertical grid line, and $p$ cut-invisible.
\end{lemma}
\begin{proof}
    The proof is similar to the one for Lemma~\ref{lem:2Co_bound}. 
    Suppose that there are two candidate triplets $(Q, o, p)$ and $(Q', o', p)$ 
    satisfying the conditions in the lemma statement.
    Assume that $y(o') < y(o)$.  
    Let $u$ be the right endpoint of the horizontal boundary cut $\cut'_{h}$ of $Q'$.
    If $p \in \oq{4}'$ then $x(p) > x(u) > x(o') = x(o)$ and 
    $y(p) < y(u) = y(o') < y(o)$. 
    Thus, $u\in\myint{R_{op}}$ and $R_{op} \not\subset Q$ 
    because $u$ is a reflex vertex of $Q$.
    This contradicts that $(Q, o, p)$ is a candidate triplet.
\end{proof}

Lemma~\ref{lem:origin_second} also holds for grid point $p$ in $\oq{2}$
and origin points lying on the same horizontal grid line. 
The number of candidate triplets of type \twori per grid line is thus $O(n^2)$ 
by the choice of $p$, which proves $\candset{\twori}= O(n^3)$.

\begin{lemma}
    \label{lem:2Ri_result}
    There are $O(n^3)$ candidate triplets of type \twori.
\end{lemma}
\begin{proof}
    We can group the candidate triplets $(Q, o, p)$ with 
    $p \in \oq{2}$ and $p$ cut-invisible in the subpolygon
    by the vertical grid lines of $\gt$ where their origin points lie.
    Similarly, we can group the candidate triplets $(Q, o, p)$ with 
    $p \in \oq{4}$ and $p$ cut-invisible in the subpolygon
    by the horizontal grid lines of $\gt$ where their origin points lie.
    By Lemma~\ref{lem:origin_second}, 
    there are $O(n^3)$ candidate triplets grouped by grid lines.
    This number is linear to the number of candidate triplets of type \twori
    because there are $O(n^2)$ grid points and $O(n)$ grid lines of $\gt$.
\end{proof}

\myparagraph{Candidate Triplets $(Q, o, p)$ with $p$ cut-visible (\tworv).}
Assume that the vertical grid lines are indexed from left to right, and 
the horizontal grid lines are indexed from bottom to top. 
Let $\gt (i, j)$ (or simply $(i, j)$ if understood in context) 
denote the grid point which is the intersection point of the $i$-th 
vertical grid line and the $j$-th horizontal grid line.
Observe that there is a unique cutset (consisting of \vcuts) that determines
$R_{op}$ for each candidate triplet $(Q, o, p)$ of type \tworv,
unless a corner of $R_{op}$ coincides with the right endpoint of $\cut_h$ 
or the upper endpoint of $\cut_v$ of $Q$.
If a corner coincides with such an endpoint of the boundary cuts,
there can be at most two cutsets, but this does not affect the asymptotic time 
complexity of the algorithm. 
So we assume that there is no such case.

Observe that $p$ has an orthogonal projection in the interior of a boundary cut,
and there is a one-to-one correspondence between candidate triplets $(Q, o, p)$ and 
uni-rectangle partitions $\ut = (Q, R_{op}, \lt)$ of type \tworv.
See Figure~\ref{fig:2R_subpolygons}(a).
Abusing the notation, let $\myink{Q, o, p} := \myink{Q, R_{op}, \lt}$.
We let $\myink{Q, o, p} = \infty$ if $p$ is not a partner point of $o$.
For a subpolygon $Q$ of type \tworv, 
we call an element $p^{*}$ of a set $S$ of grid points of $\gt$ 
an \emph{optimal partner point} of $o$ in $S$
if $p^{*} \in \arg\min_{p \in S}\myink{Q, o, p}$.

\begin{lemma}
    \label{lem:2R_observation}
    Let $Q$ and $Q'$ be subpolygons of type \twor with up-right cut direction
    and origin points $o= (i, j)$ and $o'=(i, j+1)$, respectively.
    Let $S = \{(i', j')\mid j' < j\}$ for any fixed integer $i'>i$.
    If $p$ is an optimal partner point of $o$ in $S$, 
    then $p$ is also an optimal partner point of $o'$ in $S$.
\end{lemma}
\begin{figure}[ht]
    \begin{center}
      \includegraphics[width=.9\textwidth]{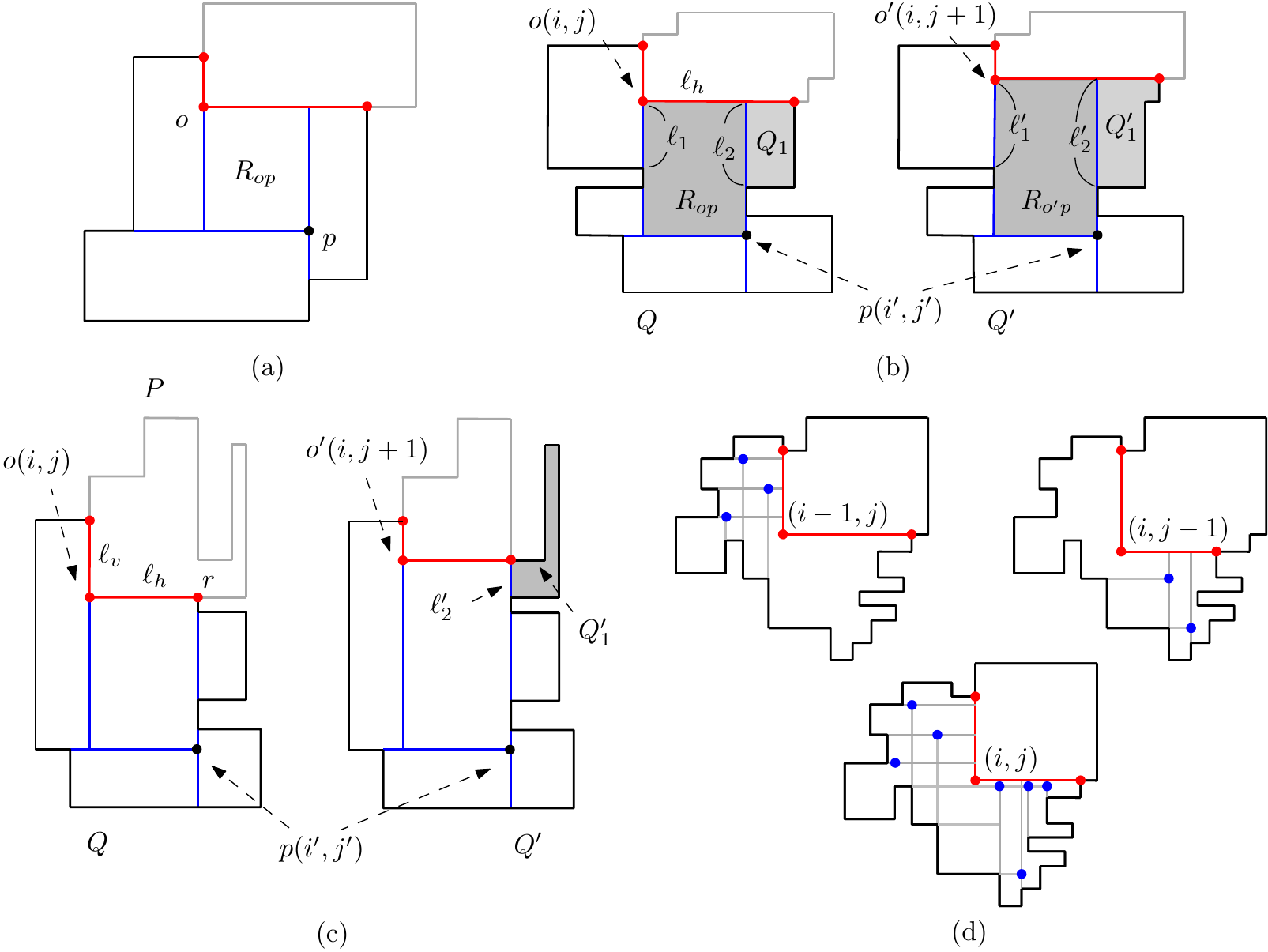}
    \end{center}
    \caption{
    (a) Unique cutset $\lt$ 
    that determines $R_{op}$ for any candidate triplet $(Q, o, p)$ of type \tworv 
    satisfying condition (2) of 
    rule 2 in Lemma~\ref{lem:minink_rules}.
    (b) $\ut' = (\ut\setminus\{Q_{1}, R_{op}\}) \cup \{Q'_{1}, R_{o'p}\}$ and 
    $\lt' = (\lt \setminus \{\cut_1, \cut_2\}) \cup \{\cut'_1, \cut'_2\}$. 
    (c) A reflex vertex $r$ of $P$ is the right endpoint of $\cut_h$ and 
the upper endpoint of a vertical edge of $Q$.
	We set $\len{\cut_2} = 0$ and $\myink{Q_1} = 0$. 
    (d) We update $A(i,j)$ using $A(i-1, j)$ and the grid points 
on the $(i-1)$-th vertical grid line, 
and then update $B(i,j)$ using $B(i, j-1)$ and the grid points
on the $(j-1)$-th horizontal grid line in $O(n)$ time. 
    }
    \label{fig:2R_subpolygons}
\end{figure}

\begin{proof}
    We show that 
\begin{equation}
\begin{gathered}
        \myink{Q, o, p} < \myink{Q, o, q} \Rightarrow \myink{Q', o', p} <= \myink{Q', o', q}, \\
        \myink{Q, o, p} = \myink{Q, o, q} \Rightarrow \myink{Q', o', p} = \myink{Q', o', q}
\end{gathered}
\end{equation}    
    for any grid points $p, q \in S$.

    Let $\ut = (Q, R_{op}, \lt)$ and $\ut' = (Q', R_{o'p}, \lt')$ 
be the uni-rectangle partitions induced by the triplet $(Q, o, p)$ and $(Q', o', p)$, respectively.
Among the cuts in $\lt$, we denote by $\cut_1$ 
the vertical cut incident to $o$, and by $\cut_2$ the vertical cut other than $\cut_1$ and incident to 
the horizontal boundary cut $\cut_h$ of $Q$.
Let $Q_{1}$ be the subpolygon in $\ut\setminus \{R_{op}\}$ 
whose boundary intersects both $\cut_2$ and $\cut_{h}$ in (nondegenerate) line segments.
Similarly, we define $\cut'_1$ and $\cut'_2$ among cuts in $\lt'$, and 
$Q'_{1}$ in $\ut'\setminus \{R_{o'p}\}$.
    See Figure~\ref{fig:2R_subpolygons}(b).

    Because $Q$, $o$ and $i'$ are fixed, 
    we can consider each of $\myink{Q, o, (i', j')}$, 
    $\len{\cut_1}$, $\len{\cut_2}$, and $\myink{Q_{1}}$
    as a (discrete) function of $j'$ with domain $\{1, ..., j-1\}$.
Similarly, we can consider each of $\myink{Q', o', (i', j')}$, 
    $\len{\cut'_1}$, $\len{\cut'_2}$, and $\myink{Q'_{1}}$
    as a (discrete) function of $j'$ with domain $\{1, ..., j-1\}$.
    
    Observe that $\ut' = (\ut\setminus\{Q_{1}, R_{op}\}) \cup \{Q'_{1}, R_{o'p}\}$ and 
    $\lt' = (\lt \setminus \{\cut_1, \cut_2\}) \cup \{\cut'_1, \cut'_2\}$. 
    So we have
    \begin{eqnarray*}
    \myink{Q', o', p} &:=& \sum_{\cut \in \lt'} \len{\cut} + \sum_{Q' \in \ut'} \myink{Q'}\\
    &=& \myink{Q, o, p} - \len{\cut_1} - \len{\cut_2} - \myink{Q_{1}} + \len{\cut'_1} + \len{\cut'_2} + \myink{Q'_{1}}\\
    &=& \myink{Q, o, p} + c.
    \end{eqnarray*} 
    as a function of $j'$ for some constant $c$.
    This is because $\len{\cut'_1} - \len{\cut_1}$,  $\len{\cut'_2} - \len{\cut_2}$, 
    $\myink{Q'_{1}}$, and $\myink{Q_{1}}$ are all constants with respect to $j'$. 
    Thus, both equations in (1) are satisfied.

    A degenerate case may occur when the right endpoint of the horizontal boundary cut $\cut_h$ of $Q$ is a reflex vertex of $P$ 
    and it coincides with the upper endpoint of a vertical edge of $Q$.
In this case,
    we set $\len{\cut_2} = 0$ and $\myink{Q_{1}} = 0$. 
    Then, both equations in (1) are satisfied. See Figure~\ref{fig:2R_subpolygons}(c).

Another degenerate case may occur when $R_{o'p} \notin Q'$ for every $p \in S$. In this case, $\myink{Q', o', p} = \infty$ for every $p \in S$, and both equations in (1) are trivially satisfied.
\end{proof}

Indeed, there may exist origin points with which no subpolygon of type \twor with 
cut direction up-right is defined. 
We can handle such cases without increasing the running time of the algorithm.

By Lemma~\ref{lem:2R_observation}, each grid point $p$ of $\gt$ 
is compared at most once as the partner point in the candidate triplets $(Q, o, p)$ 
among the subpolygons $Q$ of type \twor with up-right cut direction and
origin points lying on the same grid line.
To efficiently compute a partner point 
that minimizes $\myink{Q, o, p}$ 
for each subpolygon $Q$ of type \twor with up-right cut direction, 
we define another grid $\gtpr$ coarser than $\gt$ such that 
for every subpolygon of type \twor with up-right cut direction, 
its origin point is a grid point of $\gtpr$.
$\gtpr$ is induced by half-lines, each of which is 
a ray going either leftward horizontally or downward vertically  
from a reflex vertex of $P$.
Thus, every subpolygon of type \twor with up-right cut direction
has its origin point on a grid point of $\gtpr$.
Moreover, every grid point contained in $\myint{P}$ 
is an origin point of a subpolygon of type \twor with up-right cut direction.
The vertical grid lines of $\gtpr$ are indexed from left to right, and 
the horizontal grid lines of $\gtpr$ are indexed from bottom to top.

We maintain two arrays, $A(i,j)$ and $B(i,j)$, for each grid point $o = \gtpr(i, j)$.
$A(i,j)$ stores the partner points in quadrants $\oq{2}$ 
that minimize $\myink{Q, o(i, j), p}$, one for each horizontal grid line of $\gtpr$.
Similarly, $B(i, j)$ stores the partner points in quadrants $\oq{4}$
that minimize $\myink{Q, o(i, j), p}$, one for each vertical grid line of $\gtpr$.
We update $A(i,j)$ using $A(i-1, j)$ and the grid points 
on the $(i-1)$-th vertical grid line of $\gtpr$, 
and then update $B(i,j)$ using $B(i, j-1)$ and the grid points 
on the $(j-1)$-th horizontal grid line of $\gtpr$ in $O(n)$ time.
Then, we compute a partner point that leads to an optimal uni-rectangle partition, 
by comparing the partner points in $A(i,j)$ and $B(i,j)$ in $O(n)$ time.
See Figure~\ref{fig:2R_subpolygons}(d).
The base cases are the grid points $\gtpr(i,j)$ such that 
$\gtpr(i, j)$ is contained in $\myint{P}$ and 
both $\gtpr(i-1,j)$ and $\gtpr(i,j-1)$ are on $\mybd{P}$. 
See Figure~\ref{fig:quadratic_space}.

Since there are $O(1)$ cut directions 
and $O(n^2)$ 
grid points of $\gt$, 
we have Lemma~\ref{lem:2R_result}.
\begin{lemma}
    \label{lem:2R_result}
    We can obtain $\myink{\cdot}$ for 
    all subpolygons of type \twor satisfying rule 2 in Lemma~\ref{lem:minink_rules}, by checking 
    $O(n^3)$ uni-rectangle partitions.
\end{lemma}


\begin{figure}[ht]
    \begin{center}
      \includegraphics[width=.4\textwidth]{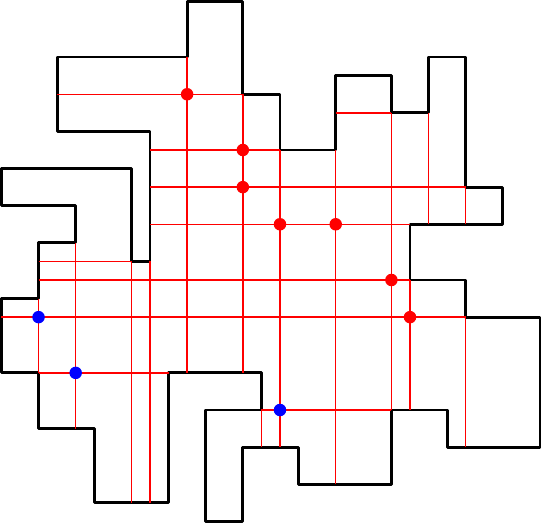}
    \end{center}
    \caption{Red grid points are on a staircase chain of grid points. Blue grid points represent base cases. 
    }
    \label{fig:quadratic_space}
\end{figure}
With the bounds on the numbers of candidate triplets of types \onec, \oner, \twoc and \twor, 
to be considered by the algorithm to obtain an \mip of $P$,
we have a main result.

\begin{theorem}
    We can compute an \mip of a rectilinear polygon with $n$ vertices and no holes in the plane in 
    $O(n^3)$ time using $O(n^2)$ space.
\end{theorem}
\begin{proof}
We analyze the time complexity of our algorithm.
The preprocessing step (which will be given in Section~\ref{sec:made_clear}) takes $O(n^2)$ time.
By Lemmas~\ref{lem:2C_result},~\ref{lem:1C_result},~\ref{lem:1R_result}, and~\ref{lem:2R_result}, 
there are $O(n^3)$ uni-rectangle partitions to be considered in order to obtain an \mip.
Since $\myink{\cdot}$ value for each uni-rectangle partition can be obtain in $O(1)$ time
after preprocessing, 
all the subpolygons encountered can be handled in $O(n^3)$ time in total.
After computing the cuts inducing an \mip of $P$,
we can construct the \mip in $O(n)$ time
as the number of cuts is $O(n)$ by Lemma~\ref{lem:vertex_cut_minink}.
So our algorithm runs in $O(n^3)$ time.

We analyze the space complexity of our algorithm. The preprocessing step for the algorithm (which will be given in Section~\ref{sec:made_clear}) uses $O(n^2)$ space.
The algorithm uses $O(1)$ space for each subpolygon of type \twoc, \onec, or \oner, 
since it stores the best partner point found so far for each origin point.
Since there are $O(n^2)$ such subpolygons, it uses $O(n^2)$ space for them in total.
Then, we analyze the space that the algorithm uses for subpolygons of type \twor.
For each origin point $\gtpr(i,j)$ of a subpolygon of type \twor, 
our algorithm uses two arrays, one for the origin point $\gtpr(i-1,j)$ and one for origin point $\gtpr(i,j-1)$.
Thus, it suffices to store the arrays corresponding to certain grid points of $\gtpr$
which are on a staircase chain 
of line segments on the grid lines of $\gtpr$. 
Since the size of each array is $O(n)$ and a staircase chain consists of 
$O(n)$ grid points at the same time,
the algorithm uses $O(n^2)$ space for subpolygons of type \twor.

When handling each subpolygon and its origin point, the algorithm stores not only $\myink{\cdot}$ value of a 
uni-rectangle partition
but also the uni-rectangle partition itself.
Because each uni-rectangle partition consists of $O(n)$ subpolygons 
with $O(n)$ vertices in total on their boundaries, 
the algorithm does not store it explicitly using a list of vertices.
Instead, it stores each uni-rectangle partition using
a pair of grid points corresponding to the origin and partner point, 
and a single number corresponding to the shape of cutsets. 
There are $O(1)$ distinct cutsets (consisting of \vcuts), 
each of which induces a uni-rectangle partition, for a given candidate triplet.
Thus, we use $O(1)$ space for storing each uni-rectangle partition.
Since the algorithm stores an optimal uni-rectangle partition for each $\le$2-cut subpolygon, 
it requires $O(n^2)$ space in total for storing them in total. 
Finally, it uses $O(n)$ space to store an \mip for $P$,
because the total number of cuts in the \mip is $O(n)$ by Lemma~\ref{lem:vertex_cut_minink}.
\end{proof}

Gonzalez and Zheng~\cite{gonzalezA} showed that the \mip algorithm by Lingas et al.~\cite{lingas}
can be extended to yield an approximation algorithm for 
the \mip partition problem on a rectangle containing $m$ point holes.
Their approximation algorithm first transforms the rectangle with point holes 
into a weakly-simple rectilinear polygon $P'$ without point holes in $O(m^2)$ time 
by connecting the point holes with layered staircase chains to the boundary of the rectangle.
Then, they apply the \mip algorithm by Lingas~et~al. to $P'$ and compute an \mip
in $O(m^4)$ time using $O(m^2)$ space. 
By applying our \mip algorithm instead of the one by Lingas~et~al., the time
complexity gets improved.

\begin{corollary}
    Given a rectangle $R$ with $m$ point holes, 
    we can compute a rectangular partition of $R$ in $O(m^3)$ time using $O(m^2)$ space
such that 
    no rectangle in the partition contains a point hole in its interior and 
    the total length of the line segments used for the partition is within three times the optimal.
\end{corollary}

\subsection{\mip algorithm made clear}
\label{sec:made_clear}
Lingas et al.~\cite{lingas} claimed that their \mip algorithm takes $O(n^4)$ time. 
The algorithm considers all candidate triplets 
for each subpolygon $Q$ of $P$. There are $O(1)$ uni-rectangle partitions 
$(Q, R_{op}, \lt)$ for a candidate triplet. 
For example, any triplet $(Q, o, p)$ of type \twoc can have at most three different cutsets,
each forming a uni-rectangle partition induced by the triplet. See Figure~\ref{fig:validcutset}
for an illustration.

\begin{figure}[ht]
    \begin{center}
      \includegraphics[width=.8\textwidth]{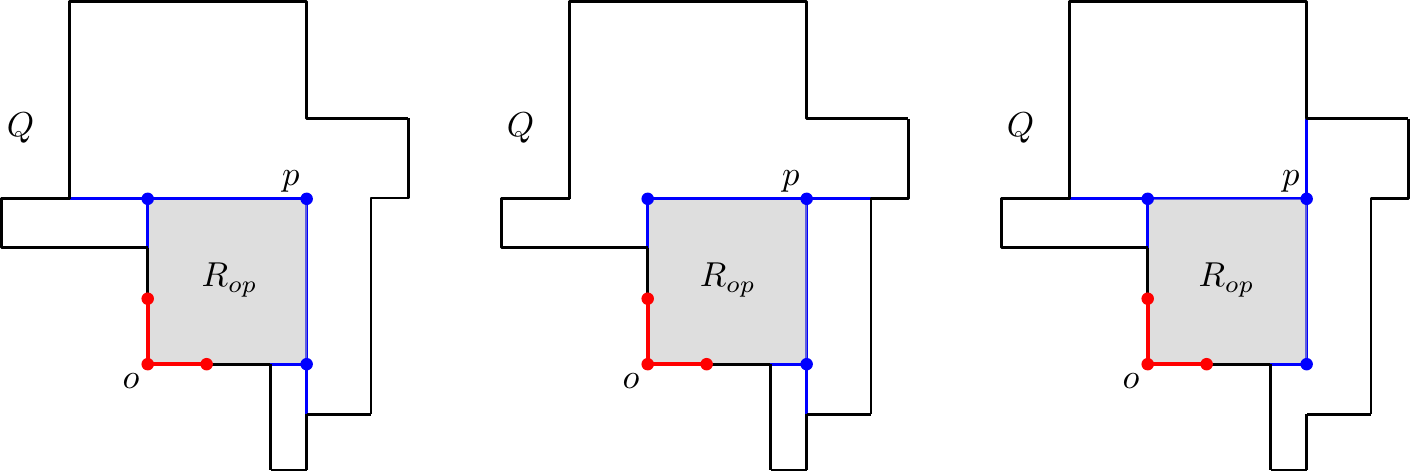}
    \end{center}
    \caption{Subpolygon $Q$ of type 2C and 
    uni-rectangle partitions induced by the triplet $(Q, o, p)$. 
    There are three distinct cutsets 
    each of which induces a uni-rectangle partition. 
    }
    \label{fig:validcutset}
\end{figure}

They claimed that given a triplet $(Q, o, p)$, they can check 
if $R_{op} \subset Q$ and compute $\myink{\ut}$ for
each of the uni-rectangle partitions $\ut = (Q, R_{op}, \lt)$ induced by the triplet
in $O(1)$ time. 
However, they did not provide any details on how to do this in $O(1)$ time.

We show that this can be done $O(1)$ time with some preprocessing.
We do not compute the uni-rectangle partitions explicitly,
because each partition may have $O(n)$ subpolygons.
Instead, we observe that there are $O(1)$ subpolygons in partition $(Q, R_{op}, \lt)$, 
each sharing a corner with $R_{op}$, and $O(n)$ subpolygons, each sharing 
a side (but sharing no corner) with $R_{op}$.
We compute $\myink{\cdot}$ for the subpolygons sharing a corner with $R_{op}$ 
in a brute-force way, and 
use range-sum queries to compute the sum of $\myink{\cdot}$ values 
for the subpolygons sharing a side with $R_{op}$.
This takes $O(1)$ time in total.

\myparagraph{Preprocessing.}
We preprocess $P$ such that given two grid points $o, p \in \gt$, 
one can determine if $R_{op}\subseteq P$ or not in $O(1)$ time. 
If $R_{op}\subseteq P$, $\myink{\ut}$ for $\ut = (Q, R_{op}, \lt)$ can also be obtained in $O(1)$ time.
We use a line sweep algorithm~\cite{vanKreveld} to construct a query structure (two matrices of size $O(n^2)$ and $O(n)$ arrays of length $O(n)$) in the preprocessing.

A $\le$2-cut subpolygon of $P$ can be specified by 
a grid point $g \in \gt$, the cut directions $d$ from $g$, and a truth value $t$. 
See Figure~\ref{fig:urp_query}(a,b).
For a 1-cut subpolygon with boundary cut $\cut$, $g$ is the right endpoint (if $\cut$ is horizontal) or the upper endpoint (if $\cut$ is vertical), and 
$d$ indicates whether $P$ lies to the right of the ray emanating from $b$ 
and containing $\cut$.
For a 2-cut subpolygon, $g$ is the endpoint of shared by the two boundary cuts, and $d$ indicates whether $g$ is a convex vertex or a reflex vertex of the subpolygon. 

\begin{figure}[ht]
    \begin{center}
      \includegraphics[width=.65\textwidth]{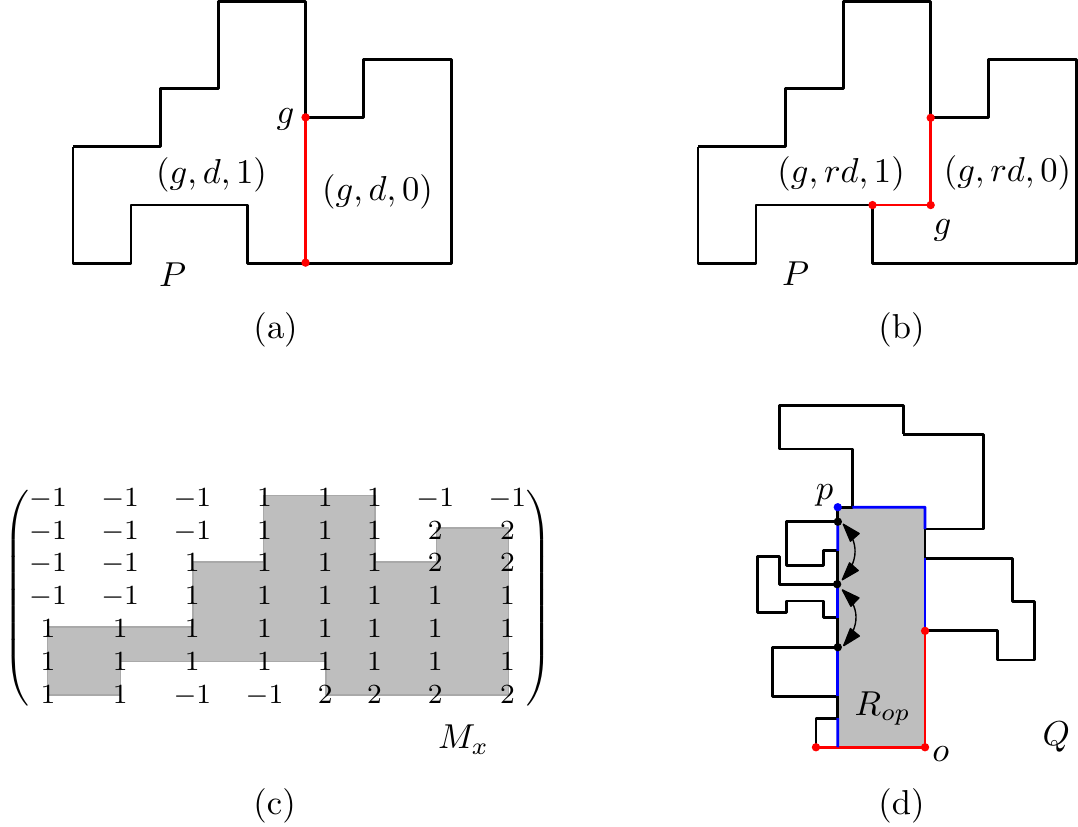}
    \end{center}
    \caption{(a) A rectilinear polygon $P$ and two 1-cut subpolygons induced by a cut incident to $g$. 
    (b) Two 2-cut subpolygons of $P$ induced by two cuts incident to $g$.
    (c) Each element of $M_x$ corresponds to a grid point of $\gt$. The gray region represents $P$.
    (d) For subpolygons of type \onec with boundary cuts lying on the same grid line of 
    $\gt$, we store pointers between two consecutive subpolygons along the grid line.
    }
    \label{fig:urp_query}
\end{figure}

With a line sweep, we construct a matrix $M_x$ ($M_y$) such that
each element corresponds to a grid point 
and two elements in the same row (column) have the same value if and only if 
the line segment connecting the grid points of $\gt$ 
corresponding to the two elements is contained in $P$.
See Figure~\ref{fig:urp_query}(c) for an illustration. 
The line sweep takes $O(n^2)$ time and space. 
Because $P$ has no holes, $R_{op}$ is contained in $P$ if and only if 
its four sides are all contained in $P$. 
By comparing the elements in $M_x$ ($M_y$) 
corresponding to the two horizontal (vertical) sides of $R_{op}$, 
we can check if $R_{op}\subset P$ in $O(1)$ time.

To compute $\myink{\rt_{\ut}}$ for an optimal refinement $\rt_{\ut}$ for 
a given uni-rectangle partition $\ut = (Q, R_{op}, \lt)$ in $O(1)$ time, 
we apply another line sweep that takes $O(n^2)$ time and space. 
During the line sweep, we consider the subpolygons of type \onec 
with boundary cuts lying on the grid line of $\gt$ containing the boundary cut of $Q$, 
and store pointers between two consecutive subpolygons along the grid line.
When a subpolygon $Q$ of type \onec is handled during the algorithm, 
the algorithm computes $\myink{\cdot}$ value of $Q$ and
$\myink{\cdot}$ values of the subpolygons linked by the pointers.
See Figure~\ref{fig:urp_query}(d). 
Then, we store $\myink{\cdot}$ values of such subpolygons in an array for each grid line of $\gt$.
Also, we store the lengths of collinear edges of $P$ in an array for each grid line of $\gt$.

Once we stored the lengths and $\myink{\cdot}$ values in arrays,
we compute $\myink{\cdot}$ values of the subpolygons
in a uni-rectangle partition $\ut = (Q, R_{op}, \lt)$ as follows.
There are $O(n)$ subpolygons that share a side (sharing no corner) with 
$R_{op}$. We compute $\myink{\cdot}$ values for these subpolygons
by performing four range-sum queries, each for a side of $R_{op}$. 
There are $O(1)$ subpolygons that share a corner with $R_{op}$.
We compute their $\myink{\cdot}$ values in a brute-force way.
Finally, we compute the total length of the line segments in $\lt$
by performing four range-sum queries to obtain the length of the 
intersection between $\mybd{P}$ and the sides of $R_{op}$.
All these queries and computations take $O(1)$ time in total. 

Observe that $\sum_{\cut \in \lt} \len{\cut}$ is the perimeter of $R_{op}$ minus $\sum_{\cut \in \mybd{P} \cap R_{op}} \len{\cut}$ plus the lengths of at most four line segments, each sharing a corner with $R_{op}$, which
can be computed in $O(1)$ time.
Because $\myink{\ut} := \sum_{\cut \in \lt} \len{\cut} + \sum_{Q' \in \ut} \myink{Q'}$ by definition,
we can compute $\myink{\ut}$ in $O(1)$ time.

\section{\tp algorithm}
\label{sec:thick_partition}

\subsection{Rectilinear polygon with no holes}
We give a \tp algorithm for a rectilinear polygon $P$ with $n$ 
vertices and no holes under the vertex incidence.
Our algorithm is based on the \tpv with $O(n^5)$ time and $O(n^4)$ space 
by O'Rourke and Tewari~\cite{oRourke}.
But with some modifications in a way similar to the one in Section~\ref{sec:minink_mine},
our algorithm returns a \vtp in $O(n^3\log^2n)$ time using $O(n^3)$ space.

Our algorithm classifies subpolygons into four types and enumerates 
the uni-rectangle partitions of each type without duplicates.
We define $\mywidth{\cdot}$ functions as follows.
\begin{itemize}
\item $\mywidth{R} := \min (a, b)$, for a rectangle $R$ with side lengths $a$ and $b$.
\item $\mywidth{\rt} := \min_{R \in \rt} \mywidth{R}$, for a rectangular partition $\rt$.
\item $\mywidth{Q}$ is the maximum of $\mywidth{\rt}$ over all rectangular partitions $\rt$
of a rectilinear polygon $Q$. 
\item $\mywidth{Q, R_{ab}, \lt} := \min_{Q' \in \ut} \mywidth{Q'}$ 
for a uni-rectangle partition $\ut = (Q, R_{ab}, \lt)$ of a rectilinear polygon $Q$.
\end{itemize}


The running time of our algorithm is proportional to the number of uni-rectangle partitions
for which we compute $\mywidth{\cdot}$ values.
The rules in Lemma~\ref{lem:minink_rules}
still hold for $\le$2-cut subpolygons, 
even if we seek for a \vtp. 
Hence, the number of candidate triplets considered by our algorithm for each $\le$2-cut subpolygon 
remains the same for each subtype.

\begin{lemma}
\label{lem:thick_subproblem_num}
    There are $O(n^3)$ candidate triplets $(Q,o,p)$ of types \onec, \oner, \twoc,
    and \twori. 
\end{lemma}

We can preprocess $P$ so that $\mywidth{\cdot}$ value and the size of a given uni-rectangle partition can be computed efficiently, as in the following lemma.
Recall that $\card{\pt}$ denotes the number of subpolygons in a partition $\pt$ of a rectilinear polygon.

\begin{lemma}
    \label{lem:fattest_query}
    Given a uni-rectangle partition $\ut = (Q, R_{op}, \lt)$, 
    we can compute $\mywidth{\ut}$ and $\card{\ut}$ in $O(\log{n})$ time,
    after preprocessing of $P$ in $O(n^2 \log{n})$ time using $O(n^2)$ space.
\end{lemma}
\begin{proof}
    We preprocess $P$ in a way similar to the one in Section~\ref{sec:made_clear}.
    But when we store $\mywidth{\cdot}$ values of the subpolygons of type \onec, 
    we use segment trees instead of arrays.
    Using segment trees (Section 9.2.2 of \cite{laaksonen20}),
    we can answer a range minimum query of $n$ numbers in $O(\log{n})$ time.
    Thus, we can obtain $\mywidth{\ut}$ in $O(\log{n})$ time by 
    applying range minimum queries on the four sides of $R_{op}$ and 
    compare the result with $\mywidth{\cdot}$ values 
of $O(1)$ subpolygons in $\ut$ sharing a corner with $R_{op}$
    in a brute-force manner.

    For each segment tree storing $\mywidth{\cdot}$ values of subpolygons,
    we maintain an array storing the number of rectangles in a \vtp of each subpolygon.
    With this, we can carry out sum queries to compute $\card{\ut}$ in constant time.

    Each segment tree can be constructed in $O(n \log{n})$ time 
    using $O(n)$ space, as $O(1)$ values are stored in each of the $O(n)$ nodes.
    For each array, we use $O(n)$ time and space.
    Since there are $O(n)$ grid lines of $\gt$, the preprocessing takes 
    $O(n^2 \log{n})$ time and $O(n^2)$ space.
\end{proof}

By Lemmas~\ref{lem:thick_subproblem_num} and~\ref{lem:fattest_query}, we have the following corollary.

\begin{corollary}
    We can compute $\mywidth{\cdot}$ values of all uni-rectangle partitions of types \onec, \oner, \twoc, and 
    \twori in $O(n^3 \log n)$ time using $O(n^2)$ space.
    \label{corollary:thick_others_result}
\end{corollary}

Our algorithm handles the uni-rectangle partitions of subtype \tworv 
using the coherence between them as in Section~\ref{subsec:2R}.
Since each candidate triplet of type \tworv induces only one uni-rectangle partition, 
we abuse the notation for a candidate triplet to represent its corresponding 
uni-rectangle partition. 
Given a subpolygon $Q$ of type \tworv, its origin point $o$, and a set S of grid points of $\gt$, 
we call an element $p^{*} \in S$ an \emph{optimal partner point} of $o$ in $S$
if $p^{*} \in \arg\max_{p \in S} \mywidth{Q, o, p}$ and 
$\card{Q, R_{op^{*}}, \lt} \le \card{Q, R_{op'}, \lt}$ for any $p' \in \arg\max_{p \in S} \mywidth{Q, o, p}$.

We show that for the subpolygons of type \tworv with up-right cut direction
and origin points lying on a same vertical grid line of $\gt$,
we can compute an optimal partner point for each origin point 
among the partner points lying on another vertical grid line of $\gt$ efficiently.

Let $Q$ be a subpolygon of type \twor with up-right cut direction and origin point $o= \gt(i, j_1)$. 
By the vertex incidence, $p = \gt(i', j')$ for $i' > i$ and $j' < j_1$ can be a partner point of 
$o$ only if there is a reflex vertex of $P$ lying on the $j_1$-th line of $\gt$ 
and there is a horizontal projection from the reflex vertex onto the left side of $R_{op}$.
Hence, in the following lemma, we let $\mywidth{Q, o, p}:= 0$
if $P$ does not contain such a reflex vertex
lying on the same horizontal grid line of $\gt$ with $p$.

\begin{lemma}
    \label{lem:thick_2R_observation}
    Let $Q$ and $Q'$ be subpolygons of type \twor with up-right cut direction and origin points $o= \gt(i, j_1)$ and $o'= \gt(i, j_2)$, respectively.
    Let $S = \{\gt(i', j')\mid j' < j_1\}$ and 
    $S' = \{\gt(i', j')\mid j' < j_2\}$ for any fixed integer $i'>i$.
    We can preprocess the uni-rectangle partition $(Q, R_{op}, \lt)$ for an optimal partner point $p$ 
    of $o$ in $S$ such that an optimal partner point of $o'$ in $S'$ can be computed in $O(k \log^2 n)$ time, 
where $k$ is the number of reflex vertices of $Q'$ that lie in 
$\oq{2}\cap\oq{3}'$, each having a horizontal projection onto the left side of $R_{o'p}$. 
\end{lemma}

\begin{figure}[ht]
    \begin{center}
    \includegraphics[width=\textwidth]{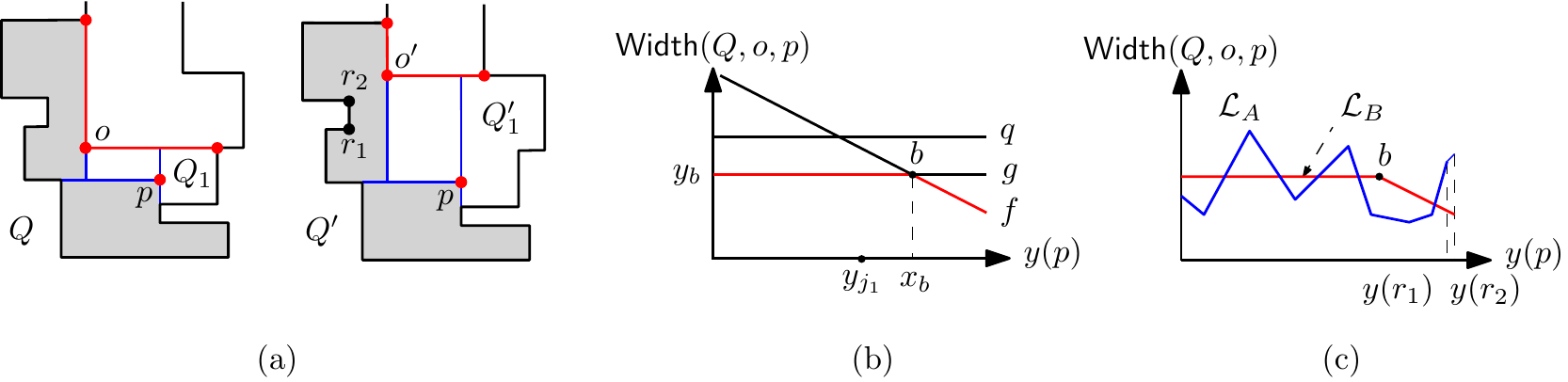}
    \caption{(a) Uni-rectangle partitions $(Q, R_{op}, \lt)$ with $o=(i,j_1)$ and 
    $(Q', R_{o'p}, \lt')$ with $o'=(i,j_2)$. The gray regions are $\qrest$. 
$r_1, r_2$ are reflex vertices contained in
$\oq{2}\cap\oq{3}'$ that have a horizontal projection onto the left side of $R_{o'p}$.
(b) Since $g$ and $q$ are constant for $y(p)$ and $f$ is decreasing linearly for $y(p)$, 
the lower envelope of the three functions has only one bending point $b$ at which its slope changes.
(c) As the origin point changes from $o$ to $o'$, 
    one can update the polygonal chain (blue graph) in $O(k \log^2 n)$ time 
    using a dynamic ray shooting data structure.
    }
    \label{fig:fattest_partition}
    \end{center}
\end{figure}

\begin{proof}
    Let $\ut = (Q, R_{op}, \lt)$ and $\ut' = (Q', R_{o'p}, \lt')$
be the uni-rectangle partitions induced by the triplets $(Q, o, p)$ and $(Q', o', p)$, respectively. 
We define $Q_{1}$ and $Q'_{1}$ as in the proof of Lemma~\ref{lem:2R_observation}.
See Figure~\ref{fig:fattest_partition}(a). 
Let $\minrest:=\min_{P' \in \qrest}\mywidth{P'}$ for $\qrest = \ut \setminus \{R_{op}, Q_1\} = \ut' \setminus \{R_{o'p}, Q'_1\}$. 
We use $x_a$ to denote the $x$-coordinate of the $a$-th vertical grid line of $\gt$,
$y_a$ to denote the $y$-coordinate of the $a$-th horizontal grid line of $\gt$, 
and $|t|$ to denote the absolute value of a real number $t$.
Then we have the following equations.
\begin{equation}
\begin{gathered}
\mywidth{Q, o, p} = 
\min_{P' \in \ut} \mywidth{P'} \notag
= \min\{|y_{j_1} - y_{j'}|, |x_{i} - x_{i'}|, \mywidth{Q_{1}}, 
\minrest\}.
\\
   \mywidth{Q', o', p} = 
\min_{P' \in \ut'} \mywidth{P'} 
= \min\{|y_{j_2} - y_{j'}|, |x_{i} - x_{i'}|, \mywidth{Q_{1}'}, \minrest\} 
\end{gathered}
\end{equation}
Observe that $\minrest$ is used in both equations. 
Thus, once we compute $\minrest$ for grid points $p\in S$, 
we use them to obtain an optimal partner point of $o'$ among grid points in $S'$.

We store $\minrest$ 
for grid points $p \in S$ in two different ways; 
as a segment tree~\cite{choi21} and as a polygonal chain. 
The segment tree is a balanced binary search tree with keys $\{1, ..., n\}$ 
at the leaf nodes. For $j' \in \{1, ..., j-1\}$, the $j'$-th leaf node stores 
$\minrest$ for $p = \gt(i', j')$. 
The leaf node also stores $\card{\qrest}$ for $p = \gt(i', j')$.
For $j' \in \{j, ..., n\}$, the $j'$-th leaf node stores the value $-\infty$.
The segment tree supports range maximum queries
$\max(a, b)$ for integers $a, b$ with $1 \le a \le b \le n$, 
which returns the index of the leaf node minimizing $\card{\qrest}$ 
among those maximizing $\minrest$ for $p=\gt(i',j')$ with $j'\in\{a,\ldots, b\}$.
If there are more than one leaf node minimizing $\card{\qrest}$ and
maximizing $\minrest$, 
the segment tree returns the index of any such leaf node.

Let $\mathcal{L}_A$ denote the graph constructed by connecting 
every two consecutive points $(y(p),\minrest)$ along the $y(p)$-axis 
by a line segment for grid points $p\in S$. See the blue graph in Figure~\ref{fig:fattest_partition}(c).

Before computing an optimal partner point of $o'$ in $S'$, 
we update $\mathcal{L}_A$ and the segment tree.
We update $\mathcal{L}_A$ by adding points $(y(p),\minrest)$
corresponding to the partner points $p \in S'\setminus S$.
We add a line segment connecting every two consecutive points.
We update the segment tree by updating leaf nodes corresponding to 
the partner points $p \in S'\setminus S$,
and by updating some internal nodes. 

Then, we obtain an optimal partner point of $o'$ in $S'$.
Because such an optimal partner $p = \gt(i', j')$ maximizes \mywidth{Q', o', p},
we can obtain it by computing the highest point on the lower envelope of 
$\mathcal{L}_A$ and $\mathcal{L}_B$,
where $\mathcal{L}_B$ denotes the lower envelope of 
the discrete functions $|y_{j_2} - y_{j'}|$, $|x_{i} - x_{i'}|$, and $\mywidth{Q_{1}'}$.
To make the computation easier, we again interpolate each of the three functions by adding edges between 
two neighboring vertices, and call them $f$, $g$, and $q$ in order. 
For example, $f$ is the interpolation of $|y_{j_2} - y_{j'}|$.

Because $g$ and $q$ are constant functions of $y(p)$ and $f$ is 
a linearly decreasing function of $y(p)$, 
The lower envelope of $f$, $g$, and $q$, denoted $\mathcal{L}_B$, has at most one \textit{bending point}, at which its slope changes.
If $\mathcal{L}_B$ has no bending point, 
then the partner point $\gt(i', j')$ with $j' =  \max(1, j_2)$ is an optimal partner point of $o'$ in $S'$.
So, suppose $\mathcal{L}_B$ has a bending point $b$.
    Denote by $j_1$ the index of the bottommost horizontal grid line of $\gt$
    among the horizontal grid lines satisfying $y_{j_1} \le x_b$. 
See Figure~\ref{fig:fattest_partition}(b).
    We first query $\max(1, j_b)$ on the segment tree.
    If  $y_b \le \minrest$    
    for the partner point $\gt(i', j')$ where $j' = \max(1, j_b)$, 
then the partner point $\gt(i', j')$ is an optimal partner point of $o'$ in $S'$.
    Otherwise, $y_b > \minrest$, 
    and we carry out a ray shooting query onto the polygonal chain 
    where the ray emanates from $b$ and proceeds along $f'$, 
    to find out the edge $e$ of the polygonal chain which is hit first by the ray.
If there is such an edge $e$, the partner point $\gt(i', j')$ with $j' = \max(1, j_{e})$ is an optimal partner point of $o'$ in $S'$, where $j_{e}$ is the index of the bottommost horizontal grid line of $\gt$
    among the horizontal grid lines whose $y$-coordinates are
    at most that of the right endpoint of $e$ on the polygonal chain.
Otherwise, there is no edge on the polygonal chain hit by the ray, 
and the partner point $\gt(i', j')$ with $j' = \max(1, j_2)$ is an optimal partner point of $o'$ in $S'$.

We now analyze the time complexity. Let $k$ denote the number of reflex vertices of $Q'$ 
that lie in  $\oq{2}\cap\oq{3}'$, each having a horizontal projection onto the left side of 
$R_{o'p}$.
It takes $O(\log^2 n)$ time to add a vertex or an edge on the polygonal chain,
using a dynamic ray-shooting data structure for connected planar subdivisions~\cite{goodrich}.
The algorithm performs $O(k)$ such operations in total, 
so the total time for updating the polygonal chain is $O(k \log^2 n)$.
We give at most one ray shooting query on the polygonal chain, which takes $O(\log^2 n)$ time.

It takes $O(\log n)$ time to update a node \ccheck{in} the segment tree~\cite{choi21}.
We update $O(k)$ nodes in total, so the total time for updating the segment tree is $O(k \log n)$.
We give at most two $\max(\cdot, \cdot)$ queries on the segment tree,
    which take $O(\log n)$ time in total. 

    Thus, in total, it takes $O(k \log^2 n)$ time to compute an optimal partner point of $o'$ in $S'$.
\end{proof}

Using coherence among the candidate triplets, 
we can compare them efficiently and obtain the following lemma. 

\begin{lemma}
    \label{lem:thick_2R_result}
    We can compute optimal partner points of $o$ for all candidate triplets $(Q, o, p)$ of type \tworv, one for each, in $O(n^3 \log^2 n)$ time using $O(n^3)$ space.
\end{lemma}
\begin{proof}
We consider only the candidate triplets of type \tworv with up-right cut direction. 
For the candidate triplets of type \tworv with other cut directions, we handle them analogously.

By Lemma~\ref{lem:thick_2R_observation}, 
for origin points lying on the same vertical grid line of $\gt$,
we can compute their optimal partner points 
among the grid points of $\gt$ on the same vertical grid line in $O(n \log^2 n)$ time, 
because there are $O(n)$ reflex vertices of $P$. 
We do this for $O(n^2)$ pairs of vertical grid lines on which origin and partner point lie,
which takes $O(n^3 \log^2 n)$ time in total.
Before computing the optimal partner points, we initialize each of the $O(n^2)$ 
segment trees such that
all leaf nodes store the value $-\infty$. It takes $O(n \log n)$ time to initialize a segment tree, thus $O(n^3 \log n)$ time in total.

A degenerate case may occur when a \emph{reflex} edge (whose both endpoints are reflex vertices of $P$) appears on the $i$-th or $i'$-th vertical grid line of $\gt$.
    In this case, we update the relevant portion of the segment tree and the polygonal chain. 
    It takes $O(n \log^2 n)$ time to update each of them, as we change $O(n)$ vertices 
    and entries in the polygonal chain and the segment tree, respectively.
    Since there are $O(n)$ reflex edges of $P$ and each reflex edge induces 
    $O(n)$ degenerate cases, we can handle all such degenerate cases
    in $O(n^3 \log^2 n)$ time in total.
    
    We maintain a polygonal chain and a segment tree for each pair of vertical grid lines
    of $\gt$ on which origin point and partner point lie.
    Because the size of each polygonal chain or segment tree is $O(n)$,
    the total space complexity is $O(n^3)$.
\end{proof}

By Corollary~\ref{corollary:thick_others_result} and Lemma~\ref{lem:thick_2R_result}, 
it takes $O(n^3 \log^2 n)$ time and $O(n^3)$ space to check all uni-rectangle partitions
and to obtain a \vtp.
After computing the cuts inducing a \vtp,
we can compute the \vtp in $O(n)$ time and space
as there are $O(n)$ cuts by Lemma~\ref{lem:vertex_cut_minink}.
So we have another main result.

\begin{theorem}
     We can compute a \vtp of a rectilinear polygon with $n$ vertices and no holes in the plane in $O(n^3 \log^2{n})$ time using $O(n^3)$ space.
\end{theorem}
 
We use \emph{\atp} to refer to a \tp whose cutset consists of \acuts.

\begin{lemma}[Theorem 6 of \cite{oRourke}]
    \label{lem:vertex_to_anchored} 
    There is an \atp of $P$ such that 
    every cut in the cutset lies on the canonical grid of $P$ 
    or at fractions $\{\frac{1}{3}, \frac{1}{2}, \frac{2}{3}\}$ 
    between two edges of $P$.
\end{lemma}

By Lemma~\ref{lem:vertex_to_anchored}, there are $O(n^2)$ points on the boundary of $P$ 
where cuts can be incident.
By treating those points as vertices, O'Rourke and Tewari obtained an $O(n^{10})$-time algorithm that computes an \atp.
By using our \vtp algorithm, we have the following result.

\begin{corollary}
    \label{corollary:anchored}
     We can compute an \atp of a rectilinear polygon with no holes in the plane in  
    $O(n^6 \log^2{n})$ time using $O(n^6)$ space.
\end{corollary}

\subsection{Rectilinear polygon with holes}
We show that the decision version of the \vtp problem for 
rectilinear polygons with holes is NP-complete. 

\myparagraph{Problem statement.} 
$\myth(P, \delta, k) :=$ Given a rectilinear polygon $P$ with $n$ vertices, including the vertices of holes, 
and a positive real value $\delta$ and a positive integer $k$,
decide whether there exists a rectangular partition $\pt$ of $P$ under the vertex incidence such that 
(1) $\mywidth{\pt} \ge \delta$ and (2) $\card{\pt} \le k$.
\medskip

We first show that \myth is in NP. For a problem instance $(P, \delta, k)$,
suppose that we are given a rectangular partition $\pt$ of $P$ as a solution set.
Because $\card{\pt} = O(n)$, we can check if $\card{\pt} \le k$ and $\mywidth{\pt} \ge \delta$ in $O(n)$ time.

To prove that \myth is NP-hard, we use a reduction from \plsat.
In \plsat, a 3-SAT formula $F$ with a variable set $X=\{x_1, ..., x_n\}$ 
and a clause set $C=\{c_1, ..., c_m\}$ is given
such that the graph $G(F)=(V,E)$ with 
$V=X\cup C$ and $E=\{(x_i, c_j)\mid x_i \text{ or } \ol{x_i} \text{ is a literal in } c_j\}$ is planar.
Here, each variable in $X$ has value either 1 (\true) or 0 (\false), and
each clause in $C$ consists of exactly three variables in $X$ (\emph{3-CNF}).

Our reduction algorithm follows the reduction scheme from \plsat to the \mip problem 
for rectilinear polygons with holes, by Lingas et al.~\cite{lingas}. 
It takes a 3-SAT formula as an input and constructs a rectilinear polygon 
such that the formula is satisfiable if and only if 
there exists a \vtp of the rectilinear polygon satisfying the conditions (1) and (2) 
in the problem statement of \myth. 
The rectilinear polygon is represented as the union of several \emph{gadgets}, each corresponding to a variable (\emph{variable} gadget, \vgadget in short), a clause (\emph{clause} gadget, \cgadget in short), or a connection between them (\emph{connection} gadget of types \emph{turn}, \emph{split}, \emph{inverter}, and \emph{phase shifter}).

In our reduction algorithm, we use all five gadget types (\vgadget, \cgadget, split, inverter, and phase shifter) used in the reduction algorithm by Lingas et al.~\cite{lingas},
and introduce one more gadget type, \emph{turn}.
Some of our gadgets have holes while those by Lingas et al. have no holes.
The holes in our gadgets force certain partitions in order to satisfy the condition (1) in the problem 
statement of \myth. 
See Figure~\ref{fig:windmill}.
Each hole is either a square of side length $\frac{\delta}{2}$ or a rectangle
with smaller side length $\frac{\delta}{2}$, 
where $\delta$ is from the problem instance $(P, \delta, k)$.
Note that all corners of a hole in $P$ are reflex vertices of $P$, 
so there is exactly one cut incident to each corner of a hole in any partition of $P$ 
induced by \vcuts. 
Thus, there are only two ways of partitioning $P$ around each square hole, each consisting
of four \vcuts: a \vwind as shown in Figure~\ref{fig:windmill}(a) or a \hwind as shown in Figure~\ref{fig:windmill}(b). The top-left corner of a square hole is incident to a vertical cut in the \vwind
while it is incident to a horizontal cut in a \hwind.
If there are two parallel cuts, each incident to a corner of the same side of a square hole, 
every resulting partition has $\mywidth{\cdot}$ value at most $\frac{\delta}{2}$
and does not satisfy the conditions (1) in the problem statement of \myth. See Figure~\ref{fig:windmill}(c).

\begin{figure}[ht]
    \begin{center}
    \includegraphics[width=.9\textwidth]{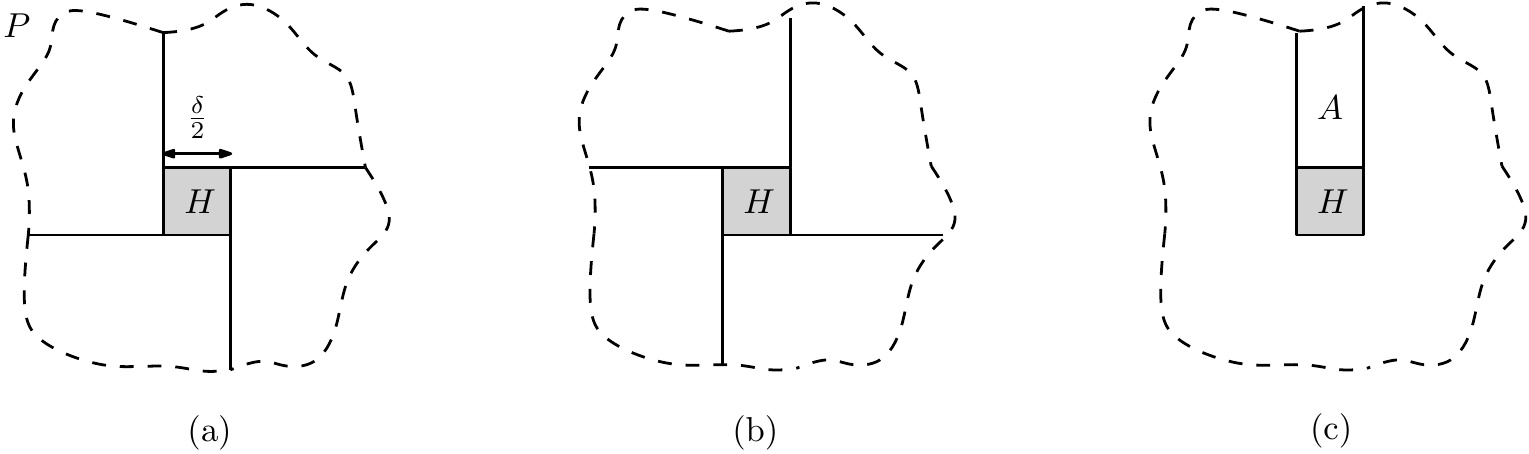}
    \caption{A rectilinear polygon $P$ and a hole $H$. $H$ is a square with side length $\frac{\delta}{2}$.
        (a) \vwind. (b) \hwind.
        (c) Two vertical cuts incident to the top corners of $H$. 
        Every resulting rectangular partition has $\mywidth{\cdot}$ value at most $\mywidth{A} \le \frac{\delta}{2}.$
    }
\label{fig:windmill}
\end{center}
\end{figure}



We now describe the shapes of our gadgets: \vgadget, four connection gadgets (turn, split, inverter, and phase shifter), and \cgadget in order.
In any gadget, every edge length is a multiple of $\frac{\delta}{2}$.

\myparagraph{\vgadget.}
A \vgadget is a polygonal chain, which is the concatenation of two staircase chains of the same direction 
with edge length $\delta$, except the two edges incident to one common corner(vertex) of the staircases
are of length $2\delta$. See Figure~\ref{fig:plsat_variable}(a).
In every polygon constructed by our reduction algorithm, a \vgadget is connected with another gadget by sharing the two endpoints of the \vgadget.

Consider the closed polygonal region bounded by a \vgadget and a horizontal line segment of length $2\delta$ which has an endpoint of the \vgadget as one of its endpoint and a (non-endpoint) vertex of the \vgadget as the other endpoint. See Figure~\ref{fig:plsat_variable}(b). There is only one way to partition the polygonal region into minimum number of rectangles, and every cut inside such a partition is a horizontal cut. 
If there is a vertical cut, the number of rectangles in the resulting partition is not the minimum. See Figure~\ref{fig:plsat_variable}(c).
Similarly, we can show that the polygonal region bounded by a \vgadget and a vertical line segment of length $2\delta$ can be partitioned into minimum number of rectangles only if the cuts are all vertical. 
See Figure~\ref{fig:plsat_variable}(d,e).
The minimum partition of the region bounded by a \vgadget and a horizontal line segment corresponds to \false, and the minimum partition of the region bounded by a \vgadget and a vertical line segment corresponds to \true.

 \begin{figure}[ht]
    \begin{center}
    \includegraphics[width=\textwidth]{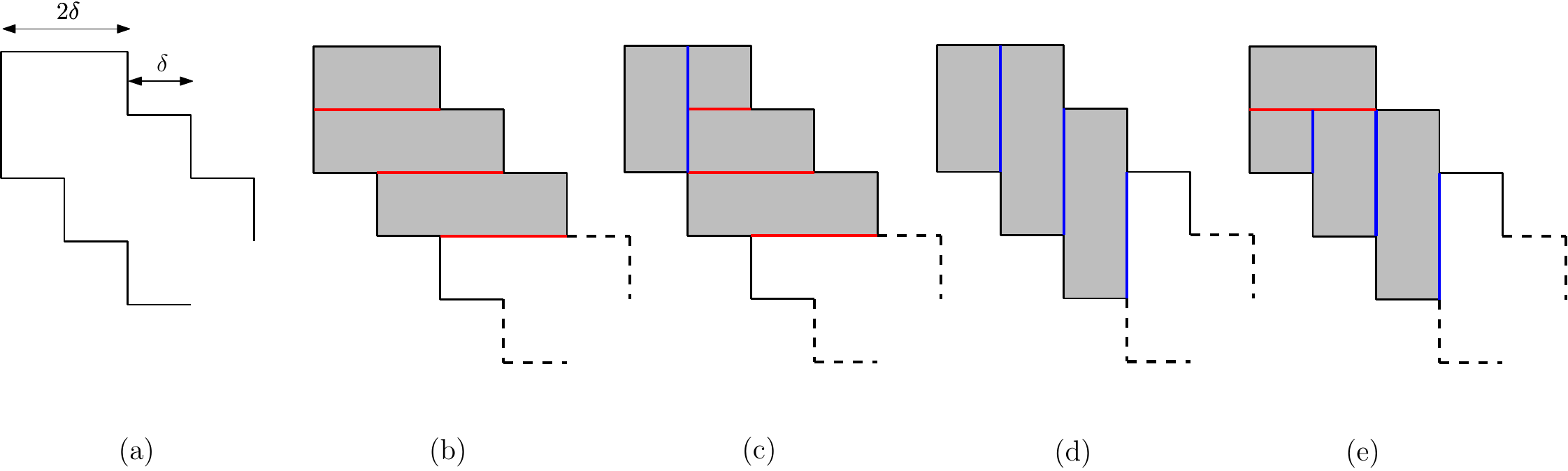}
    \caption{
(a) A \vgadget. (b) A \vgadget connected by another gadget (dashed lines). The (gray) region is bounded by a \vgadget and a horizontal line segment of length $2\delta$, and its minimum partition corresponds to \false. 
(c) If there is a vertical cut, the number of rectangles in the resulting partition is not the minimum.
(d) The (gray) region is bounded by a \vgadget and a vertical line segment of length $2\delta$, and its minimum partition corresponds to \true. 
(e) If there is a  horizontal cut, the number of rectangles in the resulting partition is not the minimum.}
    \label{fig:plsat_variable}
    \end{center}
\end{figure}

\myparagraph{Connection gadgets.}
Each connection gadget consists of at most three staircase chains and at most one square hole.
A turn gadget changes the direction (among the four diagonal directions) 
in which the truth value propagates.
A split gadget receives a \vgadget and outputs two copies of it. 
Each copied \vgadget consists of two staircase chains of the same direction, and there is no common corner of the staircases. 
An inverter gadget inverts the truth value transmitted from \true to \false, 
and vice versa.
A phase shifter gadget connects
a \vgadget and a \cgadget by matching the two endpoints of the \vgadget with two endpoints of the \cgadget (will be defined shortly), 
by using a pair of parallel edges of the same length,
which is determined by the distance between the \vgadget and the \cgadget.
See Figure~\ref{fig:gadgets} for an illustration of these gadgets.

 \begin{figure}[ht]
    \begin{center}
    \includegraphics[width=.9\textwidth]{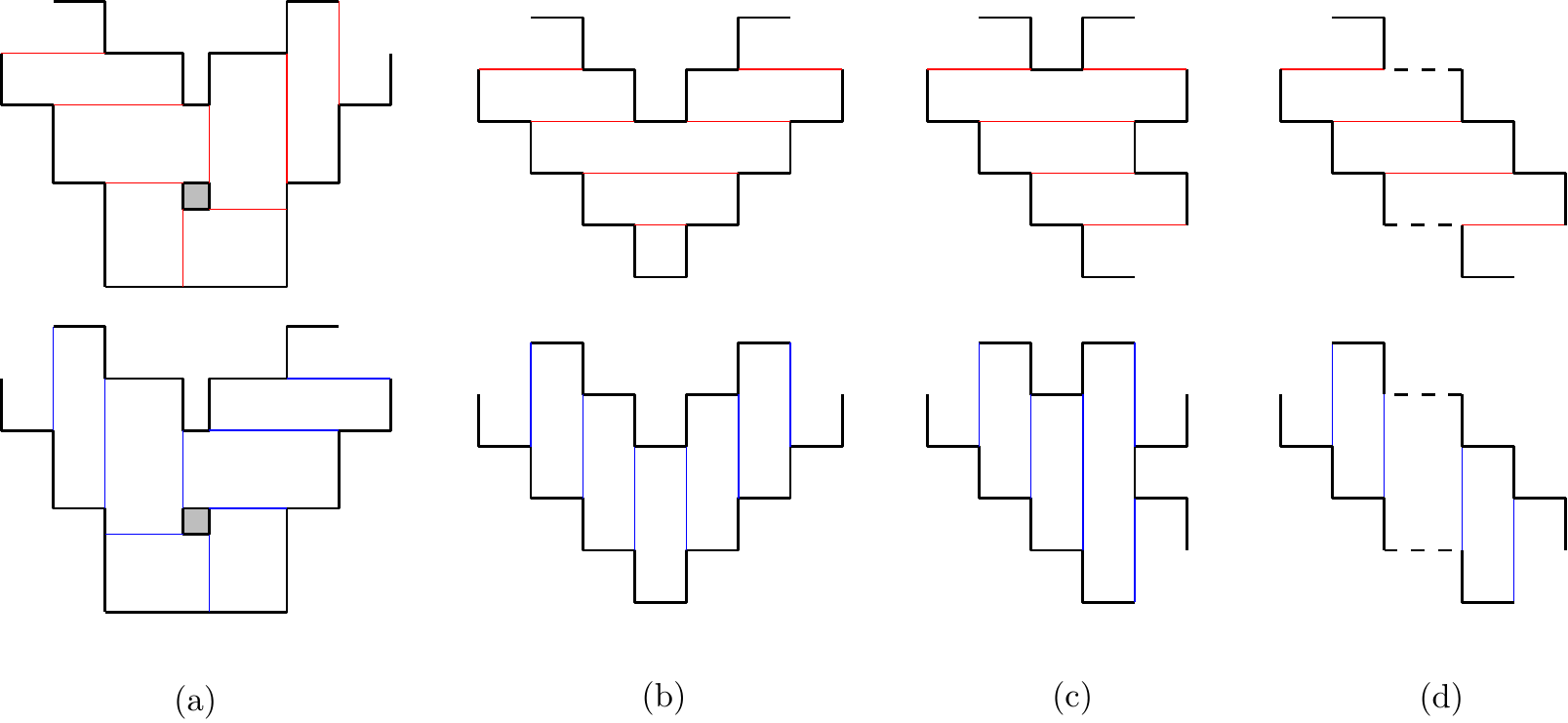}
    \caption{
    (a) An inverter gadget. The size of the square hole (gray) is $\frac{\delta}{2}$.
    (b) A turn gadget. (c) A split gadget. 
    (d) A phase shifter gadget. 
The phase shifter gadget connects a \vgadget and a \cgadget by matching the two endpoints of the \vgadget with two endpoints of the \cgadget. Dashed edges have the same length, and the length is determined by the distance between the \vgadget and the \cgadget.
}
    \label{fig:gadgets}
\end{center}
\end{figure}

\myparagraph{\cgadget.}
A \cgadget consists of three staircase chains and four holes.
A \cgadget can be connected to three \vgadgets. 
See Figure~\ref{fig:plsat_3clauses_001}(a) for an illustration of a \cgadget.
In the figure, the \cgadget is connected to \vgadgets, 
corresponding to three variables $x_1, x_2,$ and $x_3$ in the 3-SAT formula.
There are three square holes ($H_1, H_2, H_3$) with side length $\frac{\delta}{2}$ 
and a rectangle hole ($H_4$) whose smaller side length is $\frac{\delta}{2}$.
Then every partition with the $\mywidth{\cdot}$ value at least $\delta$ has a windmill pattern around each square hole.
Observe that there is no two horizontal cuts in such a partition, 
each incident to an endpoint of the same shorter side of the rectangle hole.
The distance between any point on the boundary of the \cgadget and 
any point on the holes is larger than or equal to $\delta$.
There is an inverter (including a square hole) in the top-left part of the figure, 
and it enforces that the \cgadget is not partitioned into minimum number of rectangles
if and only if $x_1 = x_2 = x_3 = 0$.
Figure~\ref{fig:plsat_3clauses_001}(b) 
shows a partition of a \cgadget when the truth values of the variables are given as $x_1=x_2=x_3=1$. 
The partition has $\mywidth{\cdot}$ value at least $\delta$ and consists of the minimum number of rectangles.
Figure~\ref{fig:plsat_3clauses_001}(c,d) 
shows two different partitions of a \cgadget when the truth values of the variables are given as $x_1=x_2=x_3=0$.
Every resulting partition either has $\mywidth{\cdot}$ value at most $\frac{\delta}{2}$ or
does not consist of the minimum number of rectangles.

\begin{figure}[ht]
    \begin{center}
    \includegraphics[width=.7\textwidth]{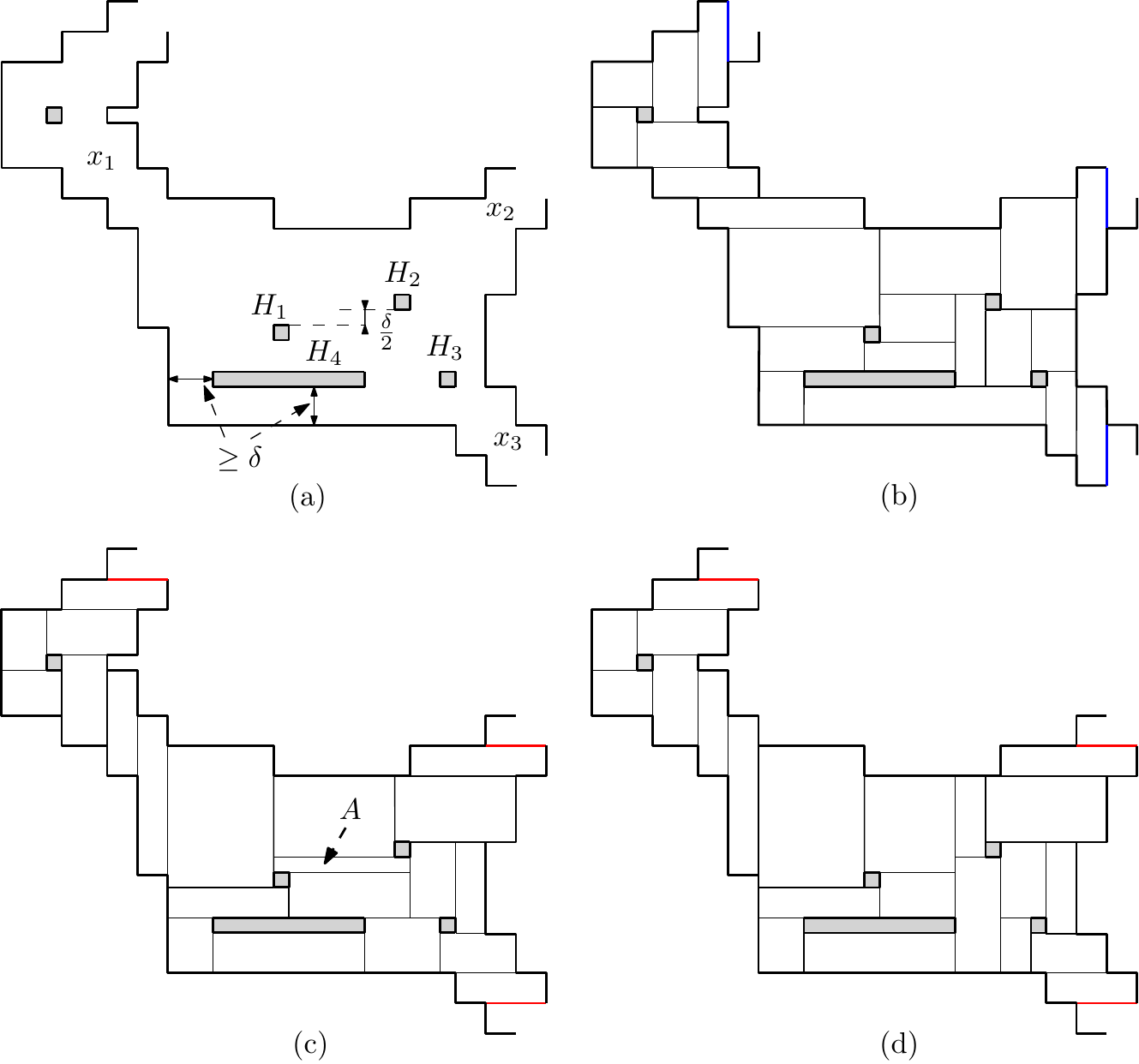}
    \caption{(a) A \cgadget corresponding to clause $(x_1\vee x_2\vee x_3)$.
    (b) A \vtp of the \cgadget with $x_1=x_2=x_3=1$. 
    (c) A partition of the \cgadget with $x_1=x_2=x_3=0$. 
    The rectangle $A$ in the middle has smaller side length $\frac{\delta}{2}$. 
    (d) Another partition of the \cgadget with $x_1=x_2=x_3=0$. 
    Every rectangle in the partition has smaller side length less than $\delta$, 
    but there are more rectangles in the partition than the minimum among such partitions.
    }
    \label{fig:plsat_3clauses_001}
    \end{center}
\end{figure}

\begin{figure}[ht]
    \begin{center}
    \includegraphics[width=.9\textwidth]{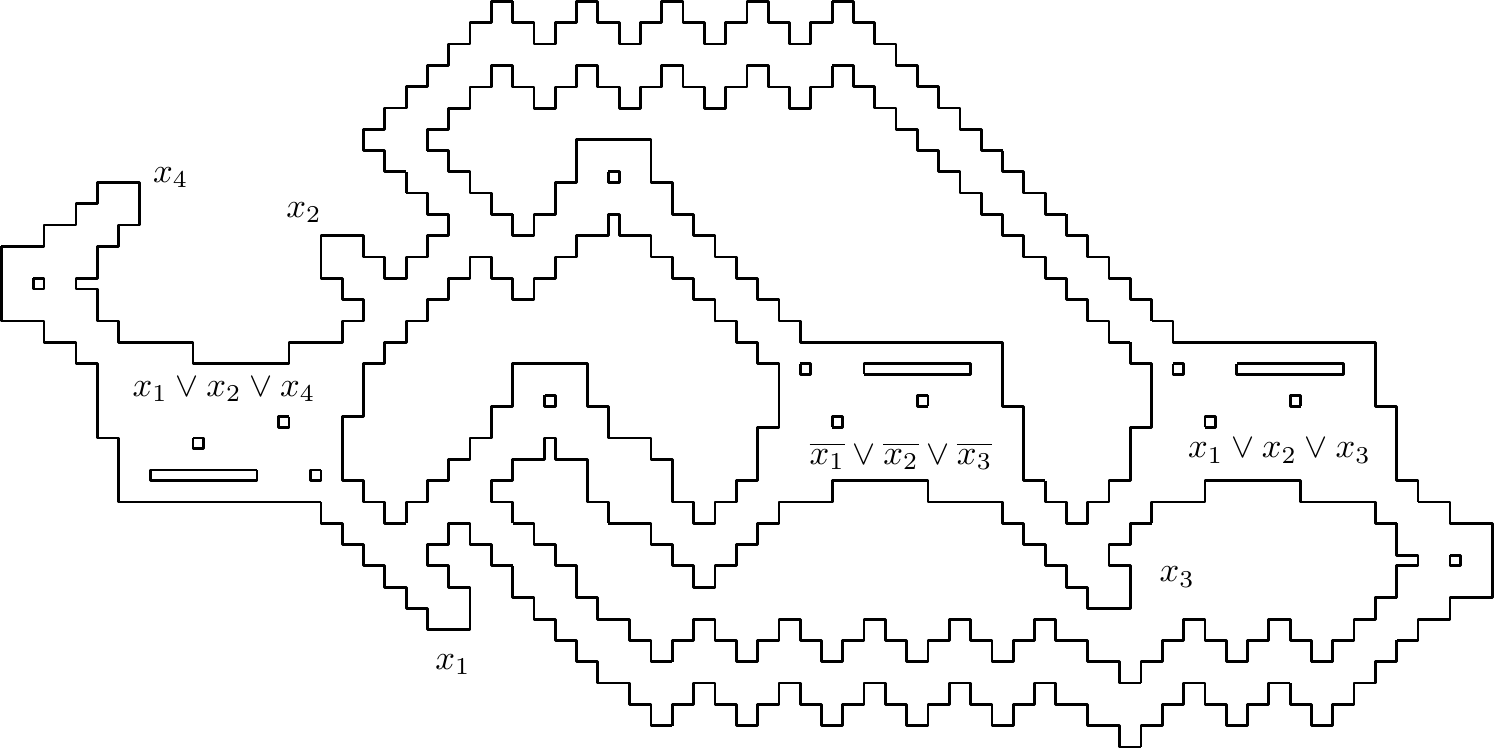}
    \caption{An example of the rectilinear polygon generated by the 3-SAT formula $(x_1\lor x_2\lor x_4)\land (\overline{x_1}\lor 
    \overline{x_2}\lor \overline{x_3})\land (x_1\lor x_2\lor x_3$).
    }
    \label{fig:plsat_ex}
    \end{center}
\end{figure}

Given a 3-SAT formula $F$, we can construct the corresponding rectilinear polygon $P$ 
in polynomial time and space,
by transforming the planar graph $G(F)$ corresponding to $F$ into a grid graph~\cite{de1990draw}.
While constructing $P$, we pick a positive real value $\delta$ such that $\mywidth{P} = \delta$.
For each gadget in $P$, we compute its \vtp 
using any polynomial time algorithm.
Let $k$ be the total number of rectangles from \vtp of each gadget. 
Then, $F$ is satisfiable if and only if the answer to $\myth(P,\delta, k)$ is \true.
Thus, \myth is NP-hard.

So we conclude with the following theorem.

 \begin{theorem}
 \label{thm:np-hardness}
    \myth is NP-complete.
 \end{theorem}

\section{Conclusion and open problems}
\label{sec:conclusion}
We gave an $O(n^3)$-time algorithm that computes an \mip of a given rectilinear polygon without holes.
We also gave an $O(n^3 \log^2 n)$-time algorithm that computes a \vtp of a given rectilinear polygon without holes.
Finally, we showed that the \vtp problem for rectilinear polygons with holes is NP-complete.
Two major open problems remain:
\begin{enumerate}
\item[A.] Does there exist a subcubic time algorithm that computes an \mip of a rectilinear polygon without holes?
\item[B.] If condition (2) is removed from \myth, is the problem still NP-hard? 
\end{enumerate}

Using coherence between uni-rectangle partitions, we could reduce the number of uni-rectangle partitions
to be checked to compute an \mip from $O(n^4)$ to $O(n^3)$, but not any further.
We conjecture that any algorithm that iterates over all uni-rectangle partitions to be checked runs in $\Omega(n^3)$ time in the worst case,
hence problem A is of great interest to us.

Our NP-hardness proof of \myth relies on both conditions (1) and (2) of \myth. 
There can be several partitions satisfying condition (1) for a given rectilinear polygon, 
but only the one correctly simulating the 3-SAT formula satisfies both (1) and (2).
So, one may need a more involved way of designing gadgets to achieve a reduction such as ours.
It would be surprising if problem B turns out to be polynomial-time solvable.

\bibliography{references.bib} 
\bibliographystyle{abbrv}
\end{document}